\documentclass[a4paper,10pt]{article}
\pdfoutput=1

\usepackage{a4wide}

\usepackage[utf8x]{inputenc}
\usepackage{ucs}
\usepackage{amsmath}
\usepackage{amsfonts}
\usepackage{amssymb}
\usepackage{amsthm}
\usepackage{algpseudocode}
\usepackage{algorithm}
\usepackage{enumitem}
\usepackage{tikz}

\algrenewcommand{\algorithmiccomment}[1]{\hfill \mbox{} \hspace*{0cm}\hfill\mbox{\small\emph{#1}}}

\renewcommand\mid{\mathrel{}\middle|\mathrel{}}
\renewcommand\({\left(}
\renewcommand\){\right)}
\renewcommand\{{\left\lbrace}
\renewcommand\}{\right\rbrace}
\newcommand\size[1]{\left\lvert#1\right\rvert}
\newcommand\TV[1]{\left\lVert#1\right\rVert_{\mathrm{TV}}}
\newcommand\N{\mathbb{N}}
\newcommand\Z{\mathbb{Z}}
\newcommand\R{\mathbb{R}}
\renewcommand\P{\mathbf{P}}
\newcommand\E[1]{\mathbf{E}\left[#1\right]}
\newcommand\Q{\mathcal{Q}}
\newcommand\A{\mathcal{A}}
\newcommand\B{\mathcal{B}}
\newcommand\J{\mathcal{J}}
\newcommand\X{\mathcal{X}}
\newcommand\G{\mathcal{G}}
\newcommand\I{\mathcal{I}}
\newcommand\1{\mathbf{1}}
\renewcommand\hat\widehat
\renewcommand\S{\mathcal{S}}
\renewcommand\bar\overline
\newcommand\inter[2]{{[\![#1,#2]\!]}}

\newcommand\word[3]{{#1_{#2 \rightarrow #3}}}
\newcommand\chain[2]{\(#1_#2\)_{#2\in\N}}
\newcommand\schain[3]{\(#1_#2\(#3\)\)_{#2\in\N}}

\newcommand\product[1]{^{\otimes #1}}
\newcommand\Act{\mathfrak A}
\newcommand\Inact{\mathfrak P}

\newcommand\draw[1]{\Call{Draw}{#1}}

\newtheorem{prop}{Property}
\newtheorem{cor}{Corollary}
\newtheorem{thm}{Theorem}
\newtheorem{define}{Definition}

\title{Speeding up Glauber Dynamics for Random Generation of Independent Sets}
\author{R\'{e}mi Varloot\thanks{Microsoft Research - Inria
Joint Centre, France.}  \and Ana Bu\v{s}i\'{c}\thanks{Inria and the Computer Science Dept.\ of \'Ecole Normale Sup\'erieure, Paris, France.}  \and Anne Bouillard\thanks{Computer Science Dept.\ of \'Ecole Normale Sup\'erieure, Paris, France.} }

\begin{document}

\maketitle

{\let\thefootnote\relax\footnotetext{The work presented in this paper has been carried out at
LINCS (www.lincs.fr) and is supported by the French National Research Agency grant ANR-12-MONU-0019.}}

\begin{abstract}

The maximum independent set (MIS) problem is a well-studied combinatorial optimization problem
that naturally arises in many applications, such as wireless communication, information theory and statistical mechanics. 

MIS problem is NP-hard, thus many results in the literature focus on fast generation of maximal independent sets of high cardinality. One possibility is to combine Gibbs sampling with coupling from the past arguments to detect convergence to the stationary regime. This results in a sampling procedure with time complexity that depends on the mixing time of the Glauber dynamics Markov chain.

We propose an adaptive method for random event generation in the Glauber dynamics that considers only the events that are effective in the coupling from the past scheme, accelerating the convergence time of the Gibbs sampling algorithm.

\end{abstract}

\section{Introduction}
\label{sec:intro}


An {\it independent set} of a graph is a set such that no two nodes in the subset are connected by an edge. 
The maximum independent set (MIS) problem is to find the set of mutually nonadjacent nodes with the largest cardinality. MIS is a well-studied combinatorial optimization problem
that naturally arises in many applications, such as statistical physics (where it is known as the hard core gas model) \cite{H04,GalTet06}, information theory \cite{GAY11} and wireless communication \cite{ShahShin12,JiangLNSW12}. 

Generating independent sets is one of the key building blocks of the wireless CSMA  \cite{GhaderiS10,SubramanianA11,ShahShin12,JiangLNSW12}. 
The interference in a wireless network can be modeled by a conflict
graph. The nodes are the links and there is an edge between two nodes if the corresponding 
communication links cannot transmit simultaneously.
At each step of the protocol a set of the communication links is chosen that forms an independent set in this conflict graph. In queue-based CSMA, the nodes have weights that are the queue sizes at the wireless links. 
Ideally, one should compute a maximum weight independent set (MWIS). 
However, MWIS problem is NP-hard  
and hard to approximate in general \cite{Trevisan01}. 
Many papers focus on finding {\it good enough} independent sets, see \cite{ShahShin12} for a more detailed discussion in the context of CSMA. 
In \cite{SSW09} the authors consider a message passing approximation algorithm for MWIS problem. They show that, if initialized using uninformative
messages, their algorithm returns an optimal value if it converges. However, the convergence is proven only for the case of a bipartite graph with a unique MWIS.


The focus of this paper is on the approximations to the MWIS problem using Glauber dynamics 
\cite{Mar99} over the space of independent sets of the interference graph. 
To simplify exposition, much of the
analysis is focused on the special case in which all the weights are equal; extensions to the
completely general case is explained In Section 4 (at the end of Subsection 4.1.). 
Glauber dynamics are defined by a (reversible) Markov process
that has as a stationary distribution a Gibbs measure
$$
\pi(A) = \frac{\lambda^{card(A)}}{Z_{\lambda}},
$$
where $A$ is an independent set of the graph, $\lambda$ is a parameter called fugacity, and $Z_{\lambda}$ is the normalization constant. 
For $\lambda = 1$, this corresponds to the uniform distribution over all independent sets and when 
$\lambda$ goes to infinity, the Gibbs measure is concentrated on the maximum independent sets.   
%
%
%
For high values of $\lambda$, the mixing time of this dynamics becomes prohibitively large. 
Furthermore, the existing bounds for the mixing time of
graphs are limited to the bounded degree case \cite{H04,JiangLNSW12,SubramanianA11}.

%



We combine the Glauber dynamics with the coupling from the past (CFTP) construction for stationarity detection of  Markovian dynamics. CFTP is an exact simulation technique introduced by Propp and Wilson \cite{PW96}. The original algorithm is computationally efficient only under some monotonicity assumptions on 
the Markovian dynamics. 
In the
general case, the CFTP algorithm requires the construction of one trajectory for each
initial condition,
which is computationally intractable in most applications.  
Huber \cite{H04} proposed a more general CFTP algorithm that is based on a construction of a bounding chain that avoids this dependence on the cardinality of the state space. 
However, this comes with a new penalty - 
the running time of the algorithm can be much larger than the mixing time of the original Markovian dynamics. Intuitively, many transitions have no effect on the bounding chain, 
inducing useless steps in the CFTP algorithm. 

The main contribution of this paper is a new CFTP algorithm that 
uses an adaptive event table and avoids generation of events that do not have any effect on the bounding chain. The idea of {\it skipping} passive events is very natural, but it is far from straightforward to see how
this can be combined with the CFTP scheme without introducing a bias. The proof that our algorithm terminates in finite expected time and provides an exact sample from the stationary distribution of the Markovian dynamics 
is stated as our main result in Theorem 2. 

We illustrate the speed-up of Glauber dynamics for the independent sets on a toy example of a star network, for which we can derive bounds for the computation time of our algorithm. We show that, unlike the initial Glauber dynamics, our algorithm does not depend on $\lambda$. 
Similar results are obtained numerically in Section~\ref{sec:numeric} for the Barab\' asi-Albert model \cite{AB02}.
We also compare the proposed algorithm against Dyer and Greenhill dynamics \cite{DG00} that use a {\it swap} operation 
to speed up the convergence.





The paper is organized as follows. An overview of the coupling from
the past construction for the exact sampling from the stationary
distribution of a Markov chain is given in Section~\ref{sec:cftp}. Our
main contribution is presented in Section~\ref{sec:skipping}. We start
in Section~\ref{ssec:oracle} by introducing the idea of active and
passive events and the construction of a dynamic event table that
contains only active events. Section~\ref{ssec:oraclecftp} contains a
detailed discussion on the skipping of events in the CFTP scheme,
summarized in Algorithm 3. The validity of the algorithm is provided
by Theorem 2. Section~\ref{sec:independent} contains the application to independent sets,
while Section 5 concludes the paper.

\section{Coupling From the Past}
\label{sec:cftp}

Throughout this section, $\S$ designates a finite state space.

We study some properties of ergodic Markov chains over $\S$, namely in the case of a joint distribution between multiple Markov chains.

\subsection{Mixing Time}

The mixing time is an indicator of how long it takes a Markov chain to forget its initial distribution. In essence, it measures how long Markov chain Monte Carlo methods must run before being ``close'' to the stationary distribution.

Let $\rho$ and $\pi$ be two probability distributions over $\S$. Recall that the \emph{total variation distance} $\TV{\cdot}$ between the two is defined as $$\TV{\rho-\pi} = \max_{A \subseteq \S} \size{\rho\(A\)-\pi\(A\)}.$$

\begin{define}[Mixing time]
Let $M$ be the transition matrix of an ergodic Markov chain over $\S$, with stationary distribution $\pi$. For all $x\in\S$, let $\schain{X}{i}{x}$ be the Markov chain with transition matrix $M$ and initial state $x$, and, for all $i\in\N$, call $\rho_i^x$ the distribution of $X_i\(x\)$. The \emph{mixing time} $t_{\mathrm{mix}}$ of $M$ is defined as $$t_{\mathrm{mix}} = \min\{ i \in \N \mid \max_{x \in \S} \TV{\rho_i^x-\pi} \leq \frac14\}.$$
\end{define}

Notice that the mixing time is defined for transition matrices rather than Markov chains, as its definition considers \emph{all} Markov chains generated by a given transition matrix.

\begin{prop} \label{thm:mix_lower_bound}
Using the above notation, if, at instant $i \in \N$, there exists an event $A$ and a state $x\in\S$ such that $$\size{\rho_i^x\(A\)-\pi\(A\)} \geq \frac{1}{4},$$ then $i \leq t_{\mathrm{mix}}$.
\end{prop}

This last property serves to derive lower-bounds on the mixing time of certain transition matrices.

\subsection{Coupling Time}

Just like the mixing time allows us to measure how quickly Markov chains converge towards their stationary distributions, the coupling time measures how long it takes before two or more Markov chains ``meet'' in a same state.

\begin{define}[Coupling]
Let $K$ be a finite set of indices. For all $k \in K$, let $$X^k = \chain{X^k}{i}$$ be a Markov chain over $\S$. A \emph{coupling} of the $X^k$ is a family of joint Markov chains $$\X = \(\chain{\X^k}{i}\)_{k \in K}$$ defined over a same probability space such that, for all $k \in K$, the marginal distribution of $\X^k = \chain{\X^k}{i}$ is that of $X^k$.
\end{define}

\begin{define}[Coupling Time]
Let $X$ be a coupling of Markov chains over $\S$. We say $X$ has \emph{coupled} at instant $i$ if all the $X^k_i$, $k \in K$ are equal.

We furthermore define the \emph{coupling time} $\tau$ of $X$ as the first instant at which $X$ has coupled, i.e. $$\tau = \min \{i\in\N\mid\forall\(k,l\)\in K^2, X^k_i = X^l_i\},$$ with the convention that $\min \emptyset = +\infty$.
\end{define}

Note that, unlike the mixing time, the coupling time is a random variable.

\subsection{Markov Automata}

Markov automata are a convenient means of writing random mapping representations \cite{LPW06} when defining a coupling between Markov chains that have the same transition matrix.

\begin{define}[Markov automaton]
A \emph{Markov automaton} is a quadruple $\A=\(\S,A,D,\cdot\)$, where:
\begin{itemize}
\item $\S$ is a finite state space;
\item $A$ is a finite alphabet;
\item $D$ is a probability distribution over $A$;
\item $\cdot$ is an action by the letters of $A$ on the states of $\S$:
$$\cdot\ : \begin{cases} \S \times A & \rightarrow \S \\ \(x,a\) & \mapsto x \cdot a \end{cases}$$
\end{itemize}
We make the assumption that, for all $a \in A$, $D\(a\) > 0$. If not, it is possible to build a reduced Markov automaton satisfying this property by removing all the letters $a \in A$ such that $D\(a\) = 0$.
\end{define}

Let $A^* = \bigcup_{k \in \N} A^k$ be the set of finite words, $A^\omega = A\product\N$ be the set of infinite words, and $A^\infty = A^\star \cup A^\omega$.

For a \emph{word} $u \in A^\infty$ and for $-\infty \leq i, j \leq \infty$, we denote $\word{u}{i}{j}$ the subword $\(u_i,\ldots,u_j\)$ if $i \leq j$, or $\epsilon$ if $j < i$.

For convenience, we furthermore write $S \cdot a$ for $\{x \cdot a \mid x \in S\}$ and $x \cdot \word{u}{1}{n}$ for $x \cdot u_1 \cdot \ldots \cdot u_n$, such that $S \cdot \word{u}{1}{n}$ stands for $$\{x \cdot u_1 \cdot \ldots \cdot u_n \mid x \in S\}.$$

Let $\A=\(\S,A,D,\cdot\)$ be a Markov automaton, and $\word{u}{1}{\infty} \sim D\product\N$. For all $x \in \S$, define $X\(x\) = \schain{X}{i}{x}$ by $$\forall i\in\N, X_i\(x\) = x \cdot \word{u}{1}{i},$$ i.e. $X_i\(x\)$ is the state reached when starting in $x$ and reading $\word{u}{1}{i}$.

\begin{prop}
For every $x \in \S$, $X\(x\)$ is a Markov chain, called the Markov chain \emph{generated} by $\A$ and $x$. Furthermore, these Markov chains have the same transition matrix $M_\A$.
\end{prop}

The family $X = \(X\(x\)\)_{x\in\S}$ is a natural coupling between these Markov chains, called the \emph{grand coupling} of $\A$.

If there exists a word $\word{u}{1}{n}$ such that $\S \cdot \word{u}{1}{n}$ is a singleton, we say that $\A$ \emph{couples}, and call $\word{u}{1}{n}$ a coupling word.

\begin{prop}[\cite{PW96}] \label{thm:couple-ergodic}
If $\A$ couples, we have that $M_\A$ is ergodic, and that $$\P\{\lim_{i\rightarrow\infty}\size{X_i} = 1\} = 1,$$ i.e. $X$ a.s. has a finite coupling time.
\end{prop}

The reciprocal is not true: it is possible to construct a non-coupling Markov automaton $\A$ such that $M_\A$ is ergodic.

If $\A$ couples, we define its stationary distribution and mixing time as those of $M_\A$, and its coupling time $\tau$ as that of $X$. The coupling time of a Markov automaton is closely linked to its mixing time, as shown in the following property.

\begin{prop}[\cite{LPW06}] \label{thm:mix-couple}
If $\A$ couples, the expected coupling time of $\A$ is lower-bounded by the mixing time $t_{\mathrm{mix}}$ of $M_\A$, i.e. $$t_{\mathrm{mix}} \leq \E{\tau}.$$
\end{prop}

It is important to underline that the distribution of the unique value of the grand coupling at the first moment of coupling is \emph{not} distributed according to the stationary distribution of $\A$ \cite{H02}. We now introduce an algorithm which uses the grand coupling of a Markov automaton to obtain that distribution.

\subsection{Coupling from the Past}

Let $\A=\(\S,A,D,\cdot\)$ be a coupling Markov automaton and $\word{u}{-\infty}{-1} \sim D\product\N$. Define $\chain{S}{i}$ by $$\forall i\in\N, S_i = \S \cdot \word{u}{-i}{-1}$$ and let $\tau^b$ be the first $i$ for which $S_i$ is a singleton. $\tau_b$ is called the backwards coupling time of the Markov automaton.

\begin{thm}[\cite{PW96}]\label{thm:cftp}
Using the above notations, we have that the unique element of $S_{\tau^b}$ is a.s. distributed according to the stationary distribution of $\A$, and that $\E{\tau^b} = \E{\tau}$, where $\tau$ is the coupling time of $\A$.
\end{thm}

This method for generating random variables according to the stationary distribution of a Markov chain is called \emph{coupling from the past} (CFTP). The sequence $\word{u}{-\infty}{-1}$ is called the \emph{generating sequence}. The corresponding algorithm is given in Algorithm \ref{alg:cftp}.

\begin{algorithm}
\caption{Coupling From the Past (CFTP)}
\label{alg:cftp}
\begin{algorithmic}
\Function{CFTP}{$\(\S,A,D,\cdot\)$}
 \For{$s \in \S$}
  \State $S\(s\) \gets s$
 \EndFor
 \Repeat
  \State $a \gets \draw{D}$
  \For{$s \in \S$}
   \State $T\(s\) \gets S\(s \cdot a\)$
  \EndFor
  \State $S \gets T$
 \Until{$\size{S\(\S\)} = 1$}
 \State \Return $\Call{UniqueElementOf}{S\(\S\protect\)}$
\EndFunction
\end{algorithmic}
\end{algorithm}

\begin{prop}
The expected complexity of the CFTP algorithm is $O\(\size{\S} \tau \gamma\)$, where $\gamma$ is the computation time of~$\cdot$.
\end{prop}

The complexity is linear in $\size{\S}$, which can be very large. A workaround for this is to use \emph{bounding chains} \cite{H04}. 

\begin{define}
Let $\B$ be a subset of the power set of $\S$, containing $\S$, and let $\circ : \B \times A \rightarrow \B$ be an operator such that for all $a \in A$ and $B \in \B$, $$x \in B \Rightarrow x \cdot a \in B \circ a.$$
The \emph{bounding chain} of $\S \cdot \word{u}{1}{n}$ induced by $\(\B, \circ\)$ is the sequence $\S \circ \word{u}{1}{n} = \S \circ u_1 \circ \ldots \circ u_n$.
\end{define}

Notice that, for any $n\in\N$ and $\word{u}{1}{n} \in A^n$, $$\S \cdot \word{u}{1}{n} \subseteq \S \circ \word{u}{1}{n},$$ hence the term \emph{bounding chain}.

Let $\word{u}{-\infty}{-1} \sim D\product\N$. For all $i\in\N$, we have that $$\S\cdot\word{u}{-i}{-1}\subseteq\S\circ\word{u}{-i}{-1}.$$ As a consequence, if there exists a word $\word{u}{-n}{-1}$ such that $\S\circ\word{u}{-n}{-1}$ is a singleton, then so is $\S\cdot\word{u}{-n}{-1}$, and they contain the same element.

From this, we derive a variant of the CFTP algorithm in which we iteratively compute $$B_i=\S\circ\word{u}{-i}{-1}$$ until we obtain a singleton. The backwards coupling time $\tau^b$ of the bounding chain is then defined as the hitting time of the set of singletons. Note that the sequence $$\(\S \circ \word{u}{-i}{-j}\)_{j\in\inter{-i}{-1}}$$ must now be recomputed at each iteration. This yields an overall complexity in $O\(\tau^2 \Gamma\)$, where $\Gamma$ represents the computation time of $\circ$, often small compared to $\gamma \size{\S}$.

The appearance of a quadratic dependency in $\tau$ can be overcome by doubling the period at each iteration \cite{PW96}. The algorithm is given in Algorithm \ref{alg:cftp-bound}, and has a complexity of $$O\(\tau \Gamma\).$$

\begin{algorithm}
\caption{CFTP with Bounding Chains}
\label{alg:cftp-bound}
\begin{algorithmic}
\Function{Bounded-CFTP}{$\(\S,A,D,\cdot\)$}
 \State $w \gets \epsilon$
 \State $k \gets 1$
 \Repeat
  \State $w \gets \draw{D\product{k}} \cdot w$
  \State $B \gets \S \circ w$
  \State $k \gets 2 k$
 \Until{$\size{B} = 1$}
 \State \Return $\Call{UniqueElementOf}{B}$
\EndFunction
\end{algorithmic}
\end{algorithm}
\section{Skipping}
\label{sec:skipping}

Skipping was introduced in \cite{PBG11}, in a form close to what we call here \emph{incremental skipping}, as a means of speeding up the CFTP algorithm by avoiding certain ``passive'' events. It is introduced here alongside our own approach, \emph{oracle skipping}.

We first introduce these two methods on forward coupling chains, and show their computational similarities. We then adapt oracle skipping to the CFTP algorithm, and give proof of its correctness.

\subsection{Oracle and Incremental Skipping}
\label{ssec:oracle}
Consider a forward coupling algorithm, with bounding chain $B = \chain{B}{i}$, and let $\word{u}{1}{\infty}$ be the corresponding sequence of letters: $$\forall i \in \N, B_i = \S \circ \word{u}{1}{i}.$$

We say a letter $a$ is \emph{inactive at instant $i$} if $B_i \circ a = B_i$, and that it is \emph{active} otherwise. Let $\Act_{\word{u}{1}{i-1}}$ be the set of active events at instant $i$. With these definitions, we consider a new bounding chain $B^\mathcal O = \chain{B^\mathcal O}{i}$ such that $$\forall i \in \N, B^\mathcal O_i = \S \circ \word{v}{1}{i},$$ with the $\word{v}{1}{\infty}$ drawn according to $D$ \emph{conditioned to being active letters}: $$\forall i \in \N, v_i \sim D\(\,\cdot \mid \Act_{\word{v}{1}{i-1}}\).$$

This method of generating a bounding chain is called \emph{oracle skipping}.

Despite oracle skipping being easy to manipulate on a theoretical level, it can be difficult to get an efficient implementation of the algorithm. This is due to the cost of computing the conditional distribution at each iteration.

The original skipping algorithm \cite{PBG11}, \emph{incremental skipping}, provides a workaround for this: rather than recomputing the entire distribution of events at each step, the algorithm updates its distribution incrementally.

Consider that the resulting bounding chain $B^\mathcal I = \chain{B^\mathcal I}{i}$ is obtained from a word $\word{v'}{1}{n}$, such that $$\forall i \in \N, B^\mathcal I_i = \S \circ \word{v'}{1}{i}.$$ Begin by setting $D^0 = D$. For all $i \in \N$, $v'_i$ is drawn according to $D^i$. If $v'_i \in \Act_{\word{v'}{1}{i-1}}$, set $D^{i+1} = D$, otherwise define $D^{i+1}$ by removing $v'_i$ from the set of possible events: $$\forall a \in A, D^{i+1}\(a\) = D^i\(a \mid a \neq v'_i \).$$ This process constructs the distributions $\(D^i\)_{i \in \N}$ recursively by removing at most one event at each iteration.

Though incremental skipping is more efficient in the case of complicated conditional distributions, the complexity of the rest of the algorithm is greater than in the case of oracle skipping. Furthermore, it is very difficult to obtain theoretical bounds with incremental skipping.

The following property justifies studying oracle skipping to derive bounds of coupling time of skipping algorithms, regardless of implementation.

\begin{prop} \label{thm:oracle}
Call $\tau_\mathcal O$ and $\tau_\mathcal I$ the coupling times of $B^\mathcal O$ and $B^\mathcal I$. We have that $$\tau_\mathcal O \leq \tau_\mathcal I \leq M \cdot \tau_\mathcal O,$$ where $M = \size{A}$ is the cardinality of the event set.
\end{prop}

\begin{proof}
Notice that $\word{v}{1}{\infty}$ can be obtained from $\word{u}{1}{\infty}$ by removing \emph{all} of its passive letters, and $\word{v'}{1}{\infty}$ can be obtained from $\word{u}{1}{\infty}$ by removing \emph{some} of its passive letters. Doing so results in a coupling of the three bounding chains $B^\mathcal O$, $B^\mathcal I$ and $B$ such that there exists increasing functions $\phi$ and $\psi$ from $\N$ to $\N$ satisfying: $$B_i^\mathcal O = B_{\phi\(i\)}^\mathcal I = B_{\psi\(\phi\(i\)\)}.$$ The first inequality is a direct consequence of this, since $\word{v}{1}{\tau_\mathcal O}$ is therefore a subsequence of $\word{v'}{1}{\tau_\mathcal I}$.

For the second inequality, notice that two active events in $\word{v'}{1}{\tau_\mathcal I}$ are separated by a sequence of pairwise different passive events, and the distance between the two is therefore at most $M$. This implies that $\word{v}{1}{\tau_\mathcal O}$ contains at least one letter out of $M$ from $\word{v'}{1}{\tau_\mathcal I}$, i.e. $$\tau_\mathcal O \geq \frac{\tau_\mathcal I}{M}.$$

This concludes the proof.
\end{proof}

\subsection{CFTP with Oracle Skipping}
\label{ssec:oraclecftp}
We now adapt oracle skipping to the CFTP algorithm. For an adaptation of incremental skipping, see \cite{PBG11}.

The difficulty in implementing a CFTP algorithm with oracle skipping lies in the fact that, as we move backwards in time, the state of the system at a \emph{fixed} instant $-k$ changes. The event $u_{-k}$ can therefore start out as active, then become passive, then active again, etc. each time we go further back in time. Whereas removing events that have become passive is not difficult, a passive event that was removed and that ought to be active once more cannot simply be pushed back in; keeping the event in memory would imply drawing every event, which defeats the purpose.

The solution adopted here consists in dropping passive letters completely, and inserting active letters according to an adequate distribution, one that preserves the dynamics of the initial bounding chain. We give the details of this algorithm, and prove its correctness.

To begin, we introduce a delimiter, denoted $\sharp$, used to split up our sequence of events, and two new operations: contraction, which consists in removing passive letters, and expansion, through which these passive letters are added back into a contracted word.

Fix a Markov automaton $\A = \(\S,A,D,\cdot\)$. Let $\B$ be a subset of the power set of $\S$ and $\circ$ be an operator, such that $\(\B, \circ\)$ induces a bounding chain for the grand coupling of $\A$.
Define the delimiter $\sharp$ as a letter that leaves states unchanged: $$ \forall x \in S, x \cdot \sharp = x  \quad \mathrm{and} \quad \forall B \in \B, B \circ \sharp = B.$$ Let $A_\sharp = A \cup \{\sharp\}$. For all $q \in [0,1]$, call $D_q$ be the distribution over $A_\sharp$ such that: $$D_q\(\sharp\) = q  \quad \mathrm{and} \quad \forall a \in A, D_q\(a\) = \(1-q\) \cdot D\(a\).$$ Furthermore, for any subset $S$ of $\S$, let $$\Act\(S\) = \{a \in A \mid S \neq S \circ a\} \cup \{\sharp\}$$ be the set of \emph{active} letters and $$\Inact\(S\) = \{a \in A \mid S = S \circ a\}$$ be the set of \emph{passive} letters. To simplify notations, we write $$\Act_\word{a}{1}{k} = \Act\(\S \circ \word{a}{1}{k}\) \quad \mathrm{and} \quad \Inact_\word{a}{1}{k} = \Inact\(\S \circ \word{a}{1}{k}\),$$ as these are the active and passive letters of the bounding chain after having read $\word{a}{1}{k}$.

Notice that $\sharp$ is an exception: it never modifies the state of the chain, yet is considered to be an active letter nonetheless.

We now define contraction and expansion. These are illustrated in Figure \ref{fig:con-exp}.

\begin{figure}
\centering
\begin{tikzpicture}[scale=0.4]
\tikzstyle{every node} = [anchor=south]
\draw (0,10) edge[->] (7.5,10);
\draw (0,10) edge[->] (0,13);
\node at (0.5,9) {\bf a}; \draw (0,11.5) edge[very thick] (1,12.5);
\node at (1.5,9) {\it b}; \draw (1,12.5) -- (2,12.5);
\node at (2.5,9) {\bf c}; \draw (2,12.5) edge[very thick] (3,10.5);
\node at (3.5,9) {\bf b}; \draw (3,10.5) edge[very thick] (4,11.5);
\node at (4.5,9) {\it b}; \draw (4,11.5) -- (5,11.5);
\node at (5.5,9) {\it b}; \draw (5,11.5) -- (6,11.5);
\node at (6.5,9) {\bf a}; \draw (6,11.5) edge[very thick] (7,10.5);
\node at (-1.1,11) {$\S\circ h$};
\node (top) at (-6,11) {\bf a \it b \bf c \bf b \it b \it b \bf a};

\draw (0,5) edge[->] (4.5,5);
\draw (0,5) edge[->] (0,8);
\node at (0.5,4) {\bf a}; \draw (0,6.5) edge[very thick] (1,7.5);
\node at (1.5,4) {\bf c}; \draw (1,7.5) edge[very thick] (2,5.5);
\node at (2.5,4) {\bf b}; \draw (2,5.5) edge[very thick] (3,6.5);
\node at (3.5,4) {\bf a}; \draw (3,6.5) edge[very thick] (4,5.5);
\node at (-1.1,6) {$\S\circ h$};
\node (mid) at (-6,6) {\bf a \bf c \bf b \bf a};

\draw (0,0) edge[->] (8.5,0);
\draw (0,0) edge[->] (0,3);
\node at (0.5,-1) {\bf a}; \draw (0,1.5) edge[very thick] (1,2.5);
\node at (1.5,-1) {\it b}; \draw (1,2.5) -- (2,2.5);
\node at (2.5,-1) {\it a}; \draw (2,2.5) -- (3,2.5);
\node at (3.5,-1) {\bf c}; \draw (3,2.5) edge[very thick] (4,0.5);
\node at (4.5,-1) {\it c}; \draw (4,0.5) -- (5,0.5);
\node at (5.5,-1) {\bf b}; \draw (5,0.5) edge[very thick] (6,1.5);
\node at (6.5,-1) {\it c}; \draw (6,1.5) -- (7,1.5);
\node at (7.5,-1) {\bf a}; \draw (7,1.5) edge[very thick] (8,0.5);
\node at (-1.1,1) {$\S\circ h$};
\node (bot) at (-6,1) {\bf a \it b \it a \bf c \it c \bf b \it c \bf a};

\draw (top.south) edge[->] node[anchor=east] {$h$-contraction} (mid.north);
\draw (mid.south) edge[->] node[anchor=east] {$(h,q)$-expansion} (bot.north);
\end{tikzpicture}

\small Active letters are shown in bold font.
\caption{Contracting and Expanding}\label{fig:con-exp}
\end{figure}
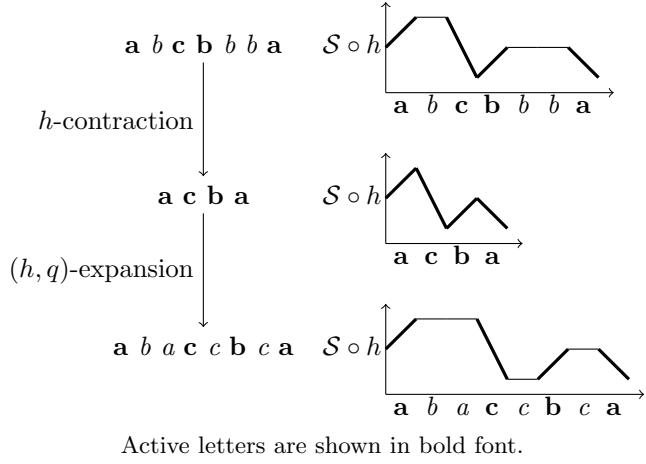

\begin{define}[$h$-contracting]
Contraction consists in removing the passive letters in a word. For a given history $h \in A_\sharp^*$, call \emph{$h$-contraction} the operation $c^h : A^\infty \rightarrow A^\infty$ defined recursively by $c^h\(\epsilon\) = \epsilon$ and $$c^h \(a \cdot u\) = \begin{cases} a \cdot c^{h \cdot a}\(u\) & \mbox{if $a \in \Act_h$} \\ c^h \( u\) & \mbox{otherwise.} \end{cases}$$
\end{define}

Notice that contraction is idempotent. A word invariant under $c^h$ is called a \emph{$h$-contracted} word.

\begin{define}[$\(h,q\)$-expanding]
Expansion consists in inserting passive letters in a word. For a given history $h \in A_\sharp^*$ and $q \in [0,1]$, call \emph{$\(h,q\)$-expansion} the operation $e^h_q : A^\infty \rightarrow A^\infty$ defined recursively by $e^h_q \(\epsilon\) = \epsilon$ and  $$e^h_q \(a \cdot u\) = \begin{cases} p \cdot e^h_q \(a \cdot u\) & \mbox{with probability $D_q\(\Inact_h\)$} \\ a \cdot e^h_q \(u\) & \mbox{with probability $D_q\(\Act_h\)$,} \end{cases}$$ where $p$ is a passive letter drawn independently according to $$D_q\(\,\cdot\mid\Inact_h\),$$ the distribution $D_q$ restricted to inactive letters.
\end{define}

Note that, during expansion, the number of passive letters inserted before each letter in the initial word is geometrically distributed.

Applying $e_q^h$ to a contracted word corresponds to constructing what the word \emph{``could have been''} before it was contracted by $c^h$, under the assumption that its letters were originally i.i.d. according to $D_q$ and that it ended with an active letter.

\begin{figure}
\centering
\begin{tikzpicture}[scale=0.4]
\tikzstyle{every node} = [anchor=south]
\draw (11,0) edge[->] (18.5,0);
\draw (11,0) edge[->] (11,3);
\node at (11.5,-1) {\bf a}; \draw (11,1.5) edge[very thick] (12,2.5);
\node at (12.5,-1) {\it b}; \draw (12,2.5) -- (13,2.5);
\node at (13.5,-1) {\bf c}; \draw (13,2.5) edge[very thick] (14,0.5);
\node at (14.5,-1) {\bf b}; \draw (14,0.5) edge[very thick] (15,1.5);
\node at (15.5,-1) {\it b}; \draw (15,1.5) -- (16,1.5);
\node at (16.5,-1) {\it b}; \draw (16,1.5) -- (17,1.5);
\node at (17.5,-1) {\bf a}; \draw (17,1.5) edge[very thick] (18,0.5);
\node at (9.9,1) {$\S\circ h$};
\node at (14.5,-1.7) {$\underbrace{\hspace*{2.8cm}}$};
\node (v) at (14.5,-2.3) {$v$};

\draw (7,-7) edge[->] (11.5,-7);
\draw (7,-7) edge[->] (7,-4);
\node at (7.5,-8) {\bf a}; \draw (7,-5.5) edge[very thick] (8,-4.5);
\node at (8.5,-8) {\bf c}; \draw (8,-4.5) edge[very thick] (9,-6.5);
\node at (9.5,-8) {\bf b}; \draw (9,-6.5) edge[very thick] (10,-5.5);
\node at (10.5,-8) {\bf a}; \draw (10,-5.5) edge[very thick] (11,-6.5);
\node at (5.9,-6) {$\S\circ h$};
\node at (9,-8.7) {$\underbrace{\hspace*{1.6cm}}$};
\node (ceq) at (9,-9.7) {$c^h\(u\)=c^h\(v\)$};
\node (heq) at (9,-12) {$u \stackrel{h}{\equiv} v$};
\draw (ceq.south) edge[->] (heq.north);

\draw (0,0) edge[->] (8.5,0);
\draw (0,0) edge[->] (0,3);
\node at (0.5,-1) {\bf a}; \draw (0,1.5) edge[very thick] (1,2.5);
\node at (1.5,-1) {\it b}; \draw (1,2.5) -- (2,2.5);
\node at (2.5,-1) {\it a}; \draw (2,2.5) -- (3,2.5);
\node at (3.5,-1) {\bf c}; \draw (3,2.5) edge[very thick] (4,0.5);
\node at (4.5,-1) {\it c}; \draw (4,0.5) -- (5,0.5);
\node at (5.5,-1) {\bf b}; \draw (5,0.5) edge[very thick] (6,1.5);
\node at (6.5,-1) {\it c}; \draw (6,1.5) -- (7,1.5);
\node at (7.5,-1) {\bf a}; \draw (7,1.5) edge[very thick] (8,0.5);
\node at (-1.1,1) {$\S\circ h$};
\node at (4,-1.7) {$\underbrace{\hspace*{3.2cm}}$};
\node (u) at (4,-2.3) {$u$};

\draw (u.south) edge[->,out=270,in=90] (9,-4);
\draw (v.south) edge[->,out=270,in=90] (9.5,-4);
\node at (9.25,-3) {$h$-contraction};
\end{tikzpicture}

\small Active letters are shown in bold font.
\caption{$h$-equivalence}\label{fig:equiv}
\end{figure}
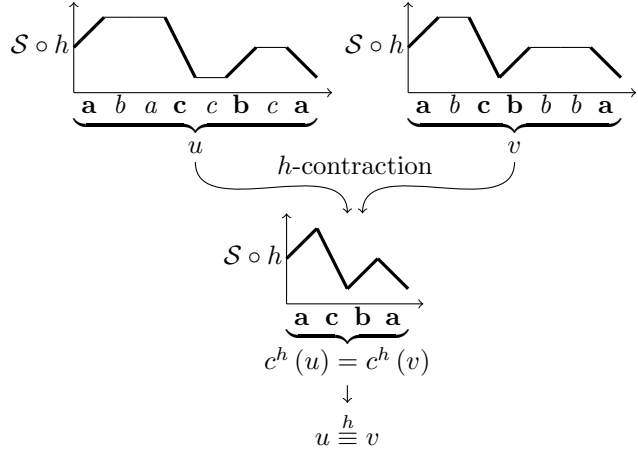

\begin{define}[$h$-equivalence]
Let $h \in A_\sharp^*$ be a history, and $u \in A_\sharp^\infty$ and $v \in A_\sharp^\infty$ be two words, possibly drawn at random. We say $u$ and $v$ are $h$-equivalent, written $u\stackrel{h}{\equiv} v$, if they almost surely have the same $h$-contractions.
\end{define}

Two words are $h$-equivalent if their contracted forms yield the same trajectories, as illustrated in Figure \ref{fig:equiv}. Since contraction removes only passive letters, the words themselves give similar trajectories, but with pauses inserted at different moments, during which the trajectory is constant.

\begin{prop} \label{thm:equivalence}
Let $h \in A_\sharp^*$ be a history, $q \in [0,1]$, and $u \in A_\sharp^\infty$. We have that
$$c^h\(u\)\stackrel{h}{\equiv}u \quad \mbox{and} \quad e_q^h\(u\)\stackrel{h}{\equiv}u,$$
and therefore that
\begin{align} e_q^h\(c^h\(u\)\)\stackrel{h}{\equiv}u. \label{eq:comp-equiv}\end{align}
Additionally, if $u \in A_\sharp^*$, $v \in A_\sharp^*$ and $u \stackrel{h}{\equiv} v$, then
\begin{align} \S \circ h \circ u = \S \circ h \circ v. \label{eq:state-equiv}\end{align}
Finally, under the same assumptions,
\begin{align} c^{h \cdot u} = c^{h \cdot v} \quad \mbox{and} \quad e_q^{h \cdot u} = e_q^{h \cdot v}. \label{eq:exp-equiv}\end{align}
\end{prop}

These results are straightforward, but will be used throughout the rest of this section to better analyse the effects of contraction and expansion. Note that the reciprocal of (2) is not true: two trajectories ending in the same state are not necessarily equivalent.

We now justify the above statement that expansion somewhat reconstructs contracted words.

\begin{prop} \label{thm:shuffling}
Let $u \sim D_q\product\N$ and $h \in A_\sharp^*$. We have that $e_q^h\(c^h\(u\)\) \sim  D_q\product\N$.
\end{prop}

\begin{proof}
Let $v = e_q^h\(c^h\(u\)\)$, $l\in\N$ and $\word{a}{1}{l} \in A_\sharp^l$. We have that $$\P\{\word{v}{1}{l} = \word{a}{1}{l}\} = \prod_{k = 1}^l \P\{v_k = a_k \mid \word{v}{1}{k-1} = \word{a}{1}{k-1}\}.$$ Showing that, for all $k\in\N$, \begin{align}\P\{v_k = a_k \mid \word{v}{1}{k-1} = \word{a}{1}{k-1}\} = D_q\(a_k\) \label{eq:indep}\end{align} would yield that $\word{v}{1}{l} \sim D_q\product{l}$. This being true for all $l\in\N$, we would in turn have that $$v \sim D_q\product{\N}.$$

We now show \eqref{eq:indep} by differentiating the two cases where $v_k$ is or is not in $\Inact_{h \cdot \word{a}{1}{k-1}}$:
\begin{align*}
& \P\{v_k = a_k \mid \word{v}{1}{k-1} = \word{a}{1}{k-1}\} \\
& \makebox[2em][c]{$=$}\ \P\{v_k = a_k \mid v_k \in \Inact_{h \cdot \word{a}{1}{k-1}} \bigwedge \word{v}{1}{k-1} = \word{a}{1}{k-1}\} \\
& \qquad \qquad \times \P\{v_k \in \Inact_{h \cdot \word{a}{1}{k-1}} \mid \word{v}{1}{k-1} = \word{a}{1}{k-1}\} \\
& \makebox[2em][r]{$+$}\ \P\{v_k = a_k \mid v_k \in \Act_{h \cdot \word{a}{1}{k-1}} \bigwedge \word{v}{1}{k-1} = \word{a}{1}{k-1}\} \\
& \qquad \qquad \times \P\{v_k \in \Act_{h \cdot \word{a}{1}{k-1}} \mid \word{v}{1}{k-1} = \word{a}{1}{k-1}\}.
\end{align*}

Notice that $v_k$ is in $\Inact_{h \cdot \word{a}{1}{k-1}}$ if and only if it was added during expansion. By definition, this occurs with probability $D_q\(\Inact_{h \cdot \word{a}{1}{k-1}}\)$, so we have that $$\P\{v_k \in \Inact_{h \cdot \word{a}{1}{k-1}} \mid \word{v}{1}{k-1} = \word{a}{1}{k-1}\} = D_q\(\Inact_{h \cdot \word{a}{1}{k-1}}\)$$ and $$\P\{v_k \in \Act_{h \cdot \word{a}{1}{k-1}} \mid \word{v}{1}{k-1} = \word{a}{1}{k-1}\} = D_q\(\Act_{h \cdot \word{a}{1}{k-1}}\)\makebox[0em]{.}$$

Consider the case where $v_k$ is passive. It was inserted during expansion, and its distribution is therefore $\left.D_q\right|_{\Inact_{h \cdot \word{a}{1}{k-1}}}$: \begin{align*} & \P\{v_k = a_k \mid v_k \in \Inact_{h \cdot \word{a}{1}{k-1}} \bigwedge \word{v}{1}{k-1} = \word{a}{1}{k-1}\} \\ & \makebox[2em][c]{$=$} D_q\(a_k \mid \Inact_{h \cdot \word{a}{1}{k-1}}\).\end{align*}

Similarly, if $v_k$ is known to be active, then it was already in $u$ and was not removed during contraction. Its distribution was therefore $D_q$ conditioned to being active, i.e. $\left.D_q\right|_{\Act_{h \cdot \word{a}{1}{k-1}}}$: \begin{align*} & \P\{v_k = a_k \mid v_k \in \Act_{h \cdot \word{a}{1}{k-1}} \bigwedge \word{v}{1}{k-1} = \word{a}{1}{k-1}\} \\ & \makebox[2em][c]{$=$} D_q\(a_k \mid \Act_{h \cdot \word{a}{1}{k-1}}\).\end{align*}

Combining all these results gives that
\begin{align*}
& \P\{v_k = a_k \mid \word{v}{1}{k-1} = \word{a}{1}{k-1}\} \\
& \makebox[2em][c]{$=$}\ D_q\(a_k \mid \Inact_{h \cdot \word{a}{1}{k-1}}\) \times D_q\(\Inact_{h \cdot \word{a}{1}{k-1}}\) \\
& \makebox[2em][r]{$+$}\ D_q\(a_k \mid \Act_{h \cdot \word{a}{1}{k-1}}\) \times D_q\(\Act_{h \cdot \word{a}{1}{k-1}}\) \\
& \makebox[2em][c]{$=$} D_q\(a_k\).
\end{align*}

This concludes the proof.
\end{proof}

Let $u \sim D_q\product\N$, and $\word{u}{1}{\sharp}$ be the same word truncated after the first appearance of the letter $\sharp$. Call $G_q$ the distribution of $\word{u}{1}{\sharp}$.

\begin{prop} \label{thm:sharp-shuffling}
Let $\word{u}{1}{\sharp} \sim G_q$ and $h \in A_\sharp^*$. We have that $e_q^h\(c^h\(\word{u}{1}{\sharp}\)\) \sim G_q$.
\end{prop}

\begin{proof}
Let $u \sim D_q\product\N$ such that $\word{u}{1}{\sharp}$ is $u$ truncated after the first $\sharp$.

By definition, $\sharp$ is always active, and is therefore neither removed when contracting nor inserted when expanding. As a result, $e_q^h\(c^h\(\word{u}{1}{\sharp}\)\)$ is $e_q^h\(c^h\(u\)\)$ truncated after the first $\sharp$. Combining this and the fact that, according to Property \ref{thm:shuffling}, $e_q^h\(c^h\(u\)\) \sim D_q\product\N$, we have that $e_q^h\(c^h\(\word{u}{1}{\sharp}\)\) \sim G_q$.
\end{proof}

This property justifies the claim that expansion corresponds to reconstructing (in distribution) a word that has been contracted, as this is indeed the case when the original word is drawn according to $G_q$.

For $n\in\N$, consider a sequence of words $\(u^m\)_{m\in\inter{1}{n}}$, independently distributed according to $G_{2^{-m}}$, and call $\G_n$ the distribution of $$u^n \cdot u^{n-1} \cdot \ldots \cdot u^1.$$

We now define the $\G$-expansion of a word. The aim is once again to rebuild what a word ``could have been'' before it was contracted, supposing it was initially drawn according to $\G_n$ for some $n\in\N$.

Formally, given a word $v$ finishing with its $n$th $\sharp$, we first split it into a sequence of words $v^n \cdot \ldots \cdot v^1$ such that each $v^m$ contains exactly one $\sharp$, which is its last letter. Notice that this decomposition is unique. We then define the $\G$-expansion of $v$ as $$e_\G\(v\) = e_{2^{-n}}^\epsilon \(v^n\) \cdot \ldots \cdot e_{2^{-m}}^{v^n \cdot \ldots \cdot v^{m+1}} \(v^m\) \cdot \ldots \cdot e_{\frac{1}{2}}^{v^n \cdot \ldots \cdot v^2} \(v^1\),$$ or $\epsilon$, if $n=0$.

\begin{prop} \label{thm:G-shuffling}
Let $u \sim \G_n$. We have that $$e_\G\(c^\epsilon\(u\)\) \sim \G_n$$ and $$e_\G\(c^\epsilon\(u\)\) \stackrel{\epsilon}{\equiv} u.$$
\end{prop}

\begin{proof}
Notice that, by definition of contraction, \begin{align} c^h\(v\cdot w\) = c^h\(v\) \cdot c^{h \cdot v}\(w\) \label{eq:c-concat} \end{align} for any finite words $v$, $w$ and $h$. Let $u^n \cdot u^{n-1} \cdot \ldots \cdot u^1$ be the unique decomposition of $u$ such that each $u^m$ ends with its unique sharp. We have that $$e_\G\(c^\epsilon\(u\)\) = e_\G\( c^\epsilon\(u^n\) \cdot \ldots \cdot c^{u^n \cdot \ldots \cdot u^2} \(u^1\) \).$$ Since each $c^{u^n \cdot \ldots \cdot u^{m+1}} \(u^m\)$ must also finish with its unique $\sharp$, the definition of $e_\G$ gives that this is equal to
\begin{align*}
& e_{2^{-n}}^\epsilon \( c^\epsilon\(u^n\) \) \cdot \ldots \cdot e_{\frac{1}{2}}^{c^\epsilon\(u^n\) \cdot \ldots \cdot c^{u^n \cdot \ldots \cdot u^3}\(u^2\)} \( c^{u^n \cdot \ldots \cdot u^2} \(u^1\) \) \\
& = e_{2^{-n}}^\epsilon \( c^\epsilon\(u^n\) \) \cdot \ldots \cdot e_{\frac{1}{2}}^{c^\epsilon\(u^n \cdot \ldots \cdot u^2\)} \( c^{u^n \cdot \ldots \cdot u^2} \(u^1\) \) \\
& = e_{2^{-n}}^\epsilon \( c^\epsilon\(u^n\) \) \cdot \ldots \cdot e_{\frac{1}{2}}^{u^n \cdot \ldots \cdot u^2} \( c^{u^n \cdot \ldots \cdot u^2} \(u^1\) \).
\end{align*}
The last line is a consequence of \eqref{eq:exp-equiv}, since $$\forall m \in \inter{1}{n},\ c^\epsilon\(u^n \cdot \ldots \cdot u^m\) \stackrel{\epsilon}{\equiv} u^n \cdot \ldots \cdot u^m.$$

Let $$v^m = e_{2^{-m}}^{u^n \cdot \ldots \cdot u^{m+1}} \( c^{u^n \cdot \ldots \cdot u^{m+1}} \(u^m\) \)$$ for all $m\in\inter{1}{n}$, such that $$e_\G\(c^\epsilon\(u\)\) = v^n \cdot \ldots \cdot v^m \cdot \ldots \cdot v^1.$$ By definition, every $u^m$ is distributed according to $G_{2^{-m}}$. Property \ref{thm:sharp-shuffling} therefore gives us that every $v^m$ is also distributed according to $G_{2^{-m}}$, which in turn implies that $v$ is distributed according to $\G_n$.

Using \eqref{eq:comp-equiv}, we also have that, for all $m \in \inter{1}{n}$, $v^m$ and $u^m$ are $\(u^n \cdot \ldots \cdot u^{m+1}\)$-equivalent, that is to say $$c^{u^n \cdot \ldots \cdot u^{m+1}} \(v^m\) = c^{u^n \cdot \ldots \cdot u^{m+1}} \(u^m\).$$ By concatenating and merging these using \eqref{eq:c-concat}, we obtain that $$c^\epsilon\(v^n \cdot \ldots \cdot v^1\) = c^\epsilon\(u^n \cdot \ldots \cdot u^1\),$$ i.e. $e_\G\(c^\epsilon\(u\)\) \stackrel{\epsilon}{\equiv} u$.
\end{proof}

Consider a sequence of words $\(u^m\)_{m\in\N}$, independently distributed such that for all $m\in\N$, $u^m \sim G_{2^{-m}}$. Define the sequence $w^n$ recursively such that $w^0 = \epsilon$ and $$w^{n+1} = c^\epsilon\(u^{n+1} \cdot e_\G\(w^n\)\).$$

This is the basis for our CFTP algorithm with oracle skipping: if $w^n$ is not a coupling word, then compute $w^{n+1}$, repeating the operation until a coupling word is found.

\begin{prop} \label{thm:inclusion}
Using the above notation, we have that, for all $m < n$, $\S \circ w^n \subseteq \S \circ w^m$. In other words, the trajectories obtained at each iteration are embedded in one another.
\end{prop}

\begin{proof}
It is enough to show that, for all $n\in\N$, $$\S \circ w^{n+1} \subseteq \S \circ w^n.$$ Since $w^{n+1}$ and $u^{n+1} \cdot e_\G\(w^n\)$ are $\epsilon$-equivalent, \eqref{eq:state-equiv} gives us that \begin{align*} \S \circ w^{n+1} & = \S \circ u^{n+1} \circ e_\G\(w^n\) \\ & \subseteq \S \circ e_\G\(w^n\), \end{align*} the inclusion being a direct consequence of the fact that $\S \circ u^{n+1} \subseteq \S$. Since $w^n$ is its own $\epsilon$-contraction, Property \ref{thm:G-shuffling} gives that $w^n$ and $e_\G\(w^n\)$ are also $\epsilon$-equivalent, and \eqref{eq:state-equiv} yields that $$\S \circ e_\G\(w^n\) = \S \circ w^n,$$ which concludes the proof.
\end{proof}

Let $$N = \inf\{\vphantom{A^A_A}n\in\N\mid\size{\S \circ w^n}=1\}$$ be the first iteration at which a coupling word is found.

\begin{thm}\label{thm:o-cftp}
Using the above notation, we have that, if there exists a coupling word for the bounding chain of $\A$, then $$\P\{N < +\infty\} = 1,$$ and $$\E{N} < +\infty,$$ i.e. the algorithm almost surely terminates, and does so after a finite expected number of iterations.

Furthermore, The unique element of $\S \circ w^N$ is distributed according to the stationary distribution $\pi$ of $\A$.
\end{thm}

Note that, due to Property \ref{thm:inclusion}, for every $n \geq N$, $\S \circ w^n$ is a singleton that contains the same state, distributed according to $\pi$.

\begin{proof}
The proof is fundamentally the same as that of Theorem \ref{thm:cftp}.

For the first part of the theorem, let $u \in A^*$ be a coupling word for the bounding chain of $\A$, and $P_u = D\product{\size{u}}\(u\)$ be the probability of drawing a word with prefix $u$. At each iteration, the probability of $u^k$ having $u$ as a prefix is the probability of there being no $\sharp$ in the first $\size{u}$ letters, i.e. having $\size{u^k} > \size{u}$, times the probability of the prefix being $u$ knowing there are no $\sharp$, i.e. $P_u$. We therefore have that
\begin{align*}
\P\{ \mbox{$u$ prefix of $u^k$}\}
& = \P\{\size{u^k} > \size{u}\} \, P_u \\
& \geq \P\{\size{u^1} > \size{u}\} \, P_u \\
& \geq \(\frac{1}{2}\)^{\size{u}} \, P_u \\
& > 0.
\end{align*}
This implies we will almost surely draw a word $u^k$ that couples, at which point $w^k$ will also couple and the algorithm will stop. Furthermore, the expected number of iterations is finite, as it is upper-bounded by $$\frac{2^{\size{u}}}{P_u}.$$

We now show that the output state is distributed according to $\pi$. Call $x_{\mbox{out}}$ the unique element of $\S \circ w^N$. If we can show that, for all $\varepsilon > 0$ and all $x \in \S$, \begin{align}\size{\P\{x_{\mbox{out}} = x\} - \pi\(x\)} \leq \varepsilon, \label{eq:TV-CFTP}\end{align}then we are finished.

Decompose $e_\G\(w^N\)$ as $$e_\G\(w^N\) = u^N \cdot \sharp \cdot u^{N-1} \cdot \sharp \cdot \ldots \cdot \sharp \cdot u^1 \cdot \sharp,$$ with $u^m \in A^*$ for all $m$. Let $\word{u^\infty}{-\infty}{0} \sim D\product{\Z^-}$, and $$\word{\widetilde w}{-\infty}{0} = \word{u^\infty}{-\infty}{0} \cdot u^N \cdot u^{N-1} \cdot \ldots \cdot u^1.$$ We begin by showing that $\word{\widetilde w}{-\infty}{0} \sim D\product{\Z^-}$.

Consider a word $\word{v}{-\infty}{0} \sim D\product{\Z^-}$ and a family of random, independent instances $\(k_m\)_{m\in\inter{1}{N}}$ such that, for each $m$, $k_m$ is distributed according to a geometric distribution of parameter $2^{-m}$. Let $l_m = \sum_{i=1}^m k_i$ for all $m\in\inter{0}{N}$.

Since $u^m\cdot\sharp \sim G_{2^{-m}}$ for all $m$, the lengths of the $u^m$ are equal to the position of the first $\sharp$ in a sequence of letters i.i.d. according to $D_{2^{-m}}\product\N$, minus one. This is exactly the geometric distribution of parameter $2^{-m}$. In particular, we have that $$u^m \sim \word{v}{-l_m+1}{-l_{m-1}},$$ and therefore
\begin{align*}
\word{\widetilde w}{-\infty}{0} & = \word{u^\infty}{-\infty}{0} \cdot u^N \cdot \ldots \cdot u^1 \\
& \sim \word{v}{-\infty}{-l_N} \cdot \word{v}{-l_N+1}{-l_{N-1}} \cdot \ldots \cdot \word{v}{-l_1+1}{-l_0} \\
& = \word{v}{-\infty}{0} \\
& \sim D\product{-\N}.
\end{align*}

Let $l=\size{u^N\cdot u^{N-1}\cdot\ldots\cdot u^1}$, such that $$\word{\widetilde w}{-l+1}{0} = u^N\cdot u^{N-1}\cdot\ldots\cdot u^1.$$ Notice that 
\begin{align*}
\{x_{\mbox{out}}\}
& = \S \circ w^N \\
& = \S \circ e_\G\(w^N\) \\
& = \S \circ \word{\widetilde w}{-l+1}{0}.
\end{align*}

We now show that if, for some $k\in\N$, $\word{\widetilde w}{-k}{0}$ is a coupling word, then $\S \circ \word{\widetilde w}{-k}{0} = \{x_{\mbox{out}}\}$.

If $k \geq l$, then 
\begin{align*}
\{x_{\mbox{out}}\}
& = \S \circ \word{\widetilde w}{-l+1}{0} \\
& = \S \circ \word{\widetilde w}{-l+1}{-k-1} \circ \word{\widetilde w}{-k}{0} \\
& \subseteq \S \circ \word{\widetilde w}{-k}{0},
\end{align*}
and if $k < l$, then
\begin{align*}
\S \circ \word{\widetilde w}{-k}{0}
& = \S \circ \word{\widetilde w}{-k}{-l} \circ \word{\widetilde w}{-l+1}{0} \\
& \subseteq \S \circ \word{\widetilde w}{-l+1}{0} \\
& = \{x_{\mbox{out}}\}.
\end{align*}
In both cases, if $\word{\widetilde w}{-k}{0}$ is a coupling word, then $\S \circ \word{\widetilde w}{-k}{0}$ is a singleton, and therefore necessarily equal to $\{x_{\mbox{out}}\}$.

We now show \eqref{eq:TV-CFTP}. Since the bounding chain couple a.s., there exists $t_\varepsilon \in \N$ such that $$\P\{\size{\S \circ {\widetilde w}_{0} \circ {\widetilde w}_{-1} \circ \ldots \circ {\widetilde w}_{-t_\varepsilon}} = 1\} \geq 1 - \varepsilon,$$ and since ${\widetilde w}_{0} \cdot {\widetilde w}_{-1} \cdot \ldots \cdot {\widetilde w}_{-t_\varepsilon}$ has the same distribution as $\word{\widetilde w}{-t_\varepsilon}{0}$, we have that $$\P\{\size{\S \circ \word{\widetilde w}{-t_\varepsilon}{0}} = 1\} \geq 1 - \varepsilon.$$

Let $Y = y \cdot \word{\widetilde w}{-t_\varepsilon}{0}$, with $y \sim \pi$. Notice that $Y \in \S \circ \word{\widetilde w}{-t_\varepsilon}{0}$, and that the letters in $\word{\widetilde w}{-t_\varepsilon}{0}$ are i.i.d. according to $D$, such that $Y\sim\pi$.

If $\size{\S \circ \word{\widetilde w}{-t_\varepsilon}{0}} = 1$, then $\S \circ \word{\widetilde w}{-t_\varepsilon}{0} = \{x_{\mbox{out}}\}$, i.e. $Y = x_{\mbox{out}}$. We therefore have that $$\P\{x_{\mbox{out}} \neq Y\} \leq \P\{\size{\S \circ \word{\widetilde w}{-t_\varepsilon}{0}} > 1\} \leq \varepsilon,$$ and thus $$\forall x \in S, \ \size{\P\{x_{\mbox{out}}=x\}-\P\{Y=x\}} \leq \varepsilon.$$ This is precisely \eqref{eq:TV-CFTP}.
\end{proof}

The aim of skipping is to keep only the contracted words, and avoid the expanded ones. Indeed, the extended words correspond to what the initial CFTP algorithm stores. Notice that when computing $w^{n+1}$, $w^n$ is expanded and then contracted. In practice, it is possible to combine these operations so as to only insert letters that are now active, rather then all ``potentially skipped'' letters. Contracting then only removes letters that were active at the previous iteration. This ensures that the size of the word (and therefore the time spent reading it) remains that of the contracted form.

We implement this algorithm using First In First Out queues to represent words. It is given in Algorithm \ref{alg:o-cftp}. Note that there is often an \emph{ad hoc} means of computing the different sets of active letters dynamically, rather than recomputing them at each iteration.

\begin{algorithm}
\caption{CFTP with Oracle Skipping}
\label{alg:o-cftp}
\begin{algorithmic}
\Function{Oracle-CFTP}{$\A = \(\S,A,D,\cdot\)$}
 \State $n \gets 0$, $w \gets []$ \Comment{$[]$ is the empty queue}
 \Repeat
  \State $\Call{Increment}{n}$
  \State $\(B,w\) \gets \Call{Double-History}{\A, w, n}$
 \Until{$\size{B} = 1$}
 \State \Return $\Call{ElementOf}{B}$
\EndFunction
\Statex
\Function{Double-History}{$\A, w, n$}
 \State $v \gets w$, $m \gets n$
 \State $\(B^+, \Act^+, u\) \gets \Call{G-Word}{\A, 2^{-m}}$ \Comment{Compute the contracted $u^n$}
 \State $B^- \gets \S$ \Comment{$B^+$ and $B^-$ are the new and old bounding chains}
 \State $\Act^- \gets \Act\(\S\)$
 \State $\Call{Decrement}{m}$
 \While{$\Call{NotEmpty}{v}$} \Comment{Expand and contract $w^{n-1}$}
  \State $\Act \gets \Act^+ \cup \Act^-$ \Comment{\textbf{Step 1}: Expanding}
  \State $a \gets \draw{\left.D_{2^{-m}}\right|_\Act}$ \Comment{Try to expend by one letter...}
  \If{$a \in \Act^-$} \Comment{... but keep it only if it is passive}
   \State $a \gets \Call{Pop}{v}$ \Comment{Otherwise, take the next letter in $v$...}
   \State $\(B^-, \Act^-\) \gets \(B^-,\Act^-\) \circ a$ \Comment{... and update the previous chain}
  \EndIf
  \If{$a \in \Act^+$} \Comment{\textbf{Step 2}: contracting}
   \State $\Call{Push}{a,u}$ \Comment{Only keep active letters...}
   \State $\(B^+, \Act^+\) \gets \(B^+,\Act^+\) \circ a$ \Comment{... and update the new chain}
   \If{$a = \sharp$}
    \State $\Call{Decrement}{m}$ \Comment{Probability of seeing $\sharp$ has changed}
   \EndIf
  \EndIf
 \EndWhile
 \State \Return $\(B^+, u\)$
\EndFunction
\Statex
\Function{G-Word}{$\A, p$}
 \State $u \gets []$ \Comment{$[]$ is the empty queue}
 \State $B \gets \S$ \Comment{Bounding chain}
 \State $\Act \gets \Act\(\S\)$
 \Repeat
  \State $a \gets \draw{\left.D_q\right|_\Act}$
  \State $\Call{Push}{a,u}$
  \State $\(B, \Act\) \gets \(B, \Act\) \circ a$
 \Until{$a = \sharp$}
 \State \Return $\(B, \Act, u\)$
\EndFunction
\end{algorithmic}
\end{algorithm}

We now give an important result for computing some upper-bounds on the computation time of our algorithm.

\begin{prop} \label{thm:o-time}
Using the previous notation, call $\tau^f_{\mathcal O}$ the coupling time of the bounding chain of $\A$, and $\tau^b_{\mathcal O} = \size{w^N}$ coupling time of the corresponding CFTP algorithm with oracle skipping. If \begin{gather} D\product{k+1}\{a_{k+1} \in \Inact_{\word{a}{1}{k}}\} \label{eq:incr} \end{gather} is increasing in $k$, i.e. passive letters become more likely as time passes, then $$\E{\tau^b_{\mathcal O}} \leq 2 \cdot \(\E{N} + \E{\tau^f_{\mathcal O}}\).$$
\end{prop}

The proof is the same as for the initial CFTP algorithm, with the following two exceptions:

\begin{itemize}
\item Since the algorithm computes transitions beyond the coupling of the bounding chain, it is important to make sure the proportion of active events after coupling does not exceed the proportion during coupling. This is ensured by condition \eqref{eq:incr}.
\item For the CFTP algorithm with oracle skipping, we introduce a delimiter $\sharp$ in our coupling word, which is not present in the initial bounding chain. This delimiter is present $N$ times in the final coupling word, hence the $\E{N}$ term to account for this.
\end{itemize}

Notice that, by construction, the number of times the algorithm goes back in time is the same as with normal CFTP. Since $\sharp$ is added exactly once every time the algorithm does so, we have that $\E{N}$ is equal to $\E{\log_2 \(\tau^b\)}$, where $\tau^b$ is the backward coupling time of the algorithm without skipping. The overall complexity of the algorithm is therefore in $O\(\E{\tau^f_{\mathcal O}}\)$, so long as this does no better than $O\(\E{\log_2 \(\tau^b\)}\)$.

This serves as a motivation to study the average forward coupling time of the Markov automaton with oracle skipping, which can serve as a means of estimating the average coupling time of the CFTP algorithm with either oracle or incremental skipping.
\section{Independent Sets}
\label{sec:independent}

Let $G=\(V,E\)$ be a simple undirected graph. A subset $I$ of $V$ is
called an independent set if no two vertices in $I$ are connected by
an edge, i.e. if $$\forall x, y \in I, \(x,y\)\notin E.$$ Let $\I$ be
the set of independent sets of $G$ and, for any vertex $v\in V$,
denote $N\(v\)$ the set of neighbors of $v$, that is to say the $w\in
V$ such that $\(v,w\)\in E$.


We study the performance of the CFTP algorithm with oracle skipping when sampling independent sets according to the distribution $$P_\lambda \(I\) = \frac{\lambda^{\size I}}{Z_\lambda}, \lambda\in\R,$$ focusing on the case where $\lambda$ is very large. Due to Property \ref{thm:o-time}, we restrict our analysis to the complexity for the forward coupling.

\subsection{Sampling algorithms}
\label{ssec:sampleindep}
We compare the coupling time our sampling algorithm with oracle
skipping with two other approaches described
in~\cite{H04}: Gibbs sampling and the Dyer-Greenhill chain \cite{DG00}.

\subsubsection{Gibbs sampling}
Let us first define a Gibbs sampler for $P_\lambda$. At each
iteration, independently draw a vertex $v$ uniformly at random and 
$u$ uniformly over $\left[0,1\right]$.
\begin{itemize}
\item If $u > \frac{\lambda}{\lambda + 1}$, then remove $v$ from $I$
  if $v\in I$, otherwise do nothing.
\item If $0 \leq u \leq \frac{\lambda}{\lambda +
    1}$, then add $v$ to $I$ if $N\(v\)\cap I = \emptyset$, otherwise do
  nothing.
\end{itemize}


This dynamic allows us to use Monte Carlo and CFTP methods to generate
independent sets according to $P_\lambda$. The CFTP approach can be
greatly improved by using the following bounding chain for the Glauber
dynamic defined in~\cite{H04}.

Consider a family of independent sets $A\subseteq \I$. Set $$B =
\cap_{I\in A} I, \qquad D=(\cup_{I\in A} I)\setminus B$$ and $$C = \cap_{I\in
  A} (V\setminus I) = V-B-D.$$  We have that $$A\subseteq \{I\in \I|
B\subseteq I\subseteq B\cup D\} = \langle B,D\rangle.$$ In other words, $B$
is the set of vertices common to every independent set in $A$, $C$ is
the set of vertices that are in none of the independent sets of $A$,
and $D$ is the set of vertices that are in some but not all of the
independent sets of $A$. The couples $\(\langle B_i,D_i\rangle\)_{i\in\N}$ 
define a bounding chain for the Glauber dynamic $\chain{A}{i}$. 

The Gibbs sampler for the bounding chain is defined as follows: at
each iteration, independently draw a vertex $v$ uniformly at random and $u$
uniformly over $\left[0,1\right]$. Suppose the initial
state is $\langle B,D\rangle$, and write $B+v$ for $B\cup\{v\}$ and $B-v$ for
$B\setminus \{v\}$; the arrival state $\langle B',D'\rangle$ is constructed as follows:
\begin{itemize}
\item If $u > \frac{\lambda}{\lambda + 1}$, then $B' = B- v$,
  $D' = D- v$.
\item If $0 \leq u \leq \frac{\lambda}{\lambda + 1}$, then:
  \begin{itemize}
   \item if $N(v) \subseteq C$, then $B' = B + v$ and $D'=D - v$,
   \item if $N(v)\cap B = \emptyset$ but $N(v) \cap D \neq \emptyset$, then $D' = D + v$ ($v$ was
    necessarily in $C\cup D$),
   \item otherwise do nothing ($v$ was necessarily
    in $C$).
  \end{itemize}
\end{itemize}

\subsubsection{The Dyer-Greenhill scheme}
The coupling time of the above bounding chain can be reduced through the
Dyer-Greenhill scheme. The main idea is to enable two elements in the
independent set to swap positions. Given $p_s\in\left[0,1\right]$, if, in
the Gibbs sampler, an attempt to add $v$ to the independent set $I$
fails due to the presence of a \emph{unique} neighbour $u$ already in $I$,
then with probability $p_s$, the independent set becomes $I+v-u$. A
bounding chain can easily be defined for this new scheme.

\subsubsection{Oracle skipping scheme}
Now consider oracle skipping for the bounding chain of the Gibbs
sampler. For each vertex $v$, we have two events: adding $v$ to $I$,
denoted $a_v$, and removing $v$ from $I$, denoted $r_v$. The active events
are:
\begin{itemize}
\item the $r_v$ for which $v \notin C$,
\item the $a_v$ for which $v \in C$ and $N\(v\) \cap B = \emptyset$,
\item the $a_v$ for which $v \in D$ and $N\(v\) \subseteq C$.
\end{itemize}


Let $V_r$ and $V_a$ be the set of vertices for which removal and addition
are respectively active in $\langle B,D\rangle$. For the Gibbs sampler,
events are drawn according to the conditional distribution by picking an
event uniformly at random in $V_z$, where $z = a$ with
probability $\lambda\size{V_a}/\(\lambda\size{V_a}+\size{V_r}\)$,
and $z = r$ otherwise.

For a vertex $v \in V$, the fact that $a_v$ and $r_v$ are active is only
modified when $v$, or one of its neighbours, is modified. It is therefore
possible to locally update the conditional distribution at each iteration
by simply updating the ``activeness'' of events for the modified vertex
and its neighbours. This justifies using oracle skipping rather than
incremental skipping in this context.

%

Note that those three samplers can be adapted to the case of weighted
vertices and product-form stationary processes of the form $$P_\Lambda
(I) = \frac{1}{Z_{\Lambda}} \prod_{v\in I} \lambda(v),$$ where
$\Lambda = (\lambda_v)_{v\in V}$ is a weight-vector of the
vertices. For the Gibbs sampler, $\lambda$ is replaced by the $\lambda(v)$
of the selected vertex. The other samplers are modified accordingly.

\subsection{Star graph}
\label{ssec:star}
\newcommand{\cP}{\mathcal{P}} 
\newcommand{\cO}{\mathcal{O}} 

In this paragraph, we study the graph $$G_n = \big(\inter{0}{n},\{(0,i),
i\in\inter{1}{n}\}\big),$$ called \emph{star graph}. We focus mainly on
the performance of the oracle skipping scheme for large values of $\lambda$,
such as when $\lambda \gg n$. The independents of this graph are
$$\I = \{\{0\}\} \cup \cP\(\inter{1}{n}\).$$

First, we consider the coupling time $\tau$ of the Glauber dynamic of
the bounding chains without skipping, both in the case of the Gibbs sampler
and of the Dyer-Greenhill sampler. 

Since at most one vertex is removed from $D$ at each iteration, and the
algorithm finishes when $D = \emptyset$, this coupling time is lower
bounded by the hitting time of
$$\langle B,\{0\}\rangle \cup \langle B,D\rangle,\ 0 \notin D.$$
Furthermore, since no vertex can be added to $B$ so long as $D$ contains
both $0$ and an element in $\inter{1}{n}$, $B = \emptyset$ until one
of those states is reached.

In the case of the Gibbs sampler, if $\lambda >
1$, the expected hitting time of $\langle \emptyset,\{0\}\rangle$ is
$O(\lambda^n)$. Furthermore, before reaching this state, the probability
of removing $0$ from $D$ is exactly $\frac{1}{(\lambda+1)(n+1)}$ at
each time step. For $n$ large enough, this gives
$$\E{\tau} \geq (n+1)(\lambda+1).$$

For the Dyer-Greenhill sampler, the coupling time $\tau^{DG}$ is greatly
reduced: the
expected hitting time of $\langle B,D\rangle$,
$0\notin D$ is constant: $$E = \frac{\lambda+1}{\lambda}\frac{n+1}{n}\frac 1 {p_s}.$$
This is due to the fact that the first attempt to swap a vertex other
than $0$ will immediately remove $0$ from $D$, since it is the only
neighbour of the selected vertex.

On the other hand, for the bounding chain to couple, every vertex must
be selected at least once for addition or removal. As at
each step, the modified vertex is chosen uniformly and independently at
random, this gives that $$\E{\tau^{DG}} \geq n\ln n + O(1).$$

Now, let us consider the (forward) coupling time $\tau^\cO$ of the
coupling chain with oracle skipping. The coupling time is at most the hitting
time of $\langle B, \emptyset \rangle$. We have two main steps to consider: 
\begin{enumerate}
\item The hitting time of a state $\langle B, D\rangle$ where $0 \notin D$;
\item The hitting time of $\langle B, \emptyset \rangle$. 
\end{enumerate}


Let us first focus on the hitting time of $\langle B, D\rangle$ where $0\notin
D$. Consider the following birth-and-death
process on $\inter{0}{n}$, where state $i$ represents the
cardinal of $C$, until $0$ is added to $C$. In state $i$, the active
events are the addition of vertices in $C$ and the removal of vertices
in $D$. As a consequence, the probabilities $p_{i,i+1}$ and
$p_{i+1,i}$ to go respectively from state $i$ to state $i+1$ and from
$i+1$ to $i$ are 
\begin{equation}
\label{eq:pi}
p_{i,i+1} = \frac{n-i}{n-i + i\lambda} \text{ and }
p_{i+1,i} = \frac{(i+1)\lambda}{n-i-1 + (i+1)\lambda}.\
\end{equation}

As we assumed $\lambda \geq n$, 
computations show that the stationary distribution $\pi$ of this birth-and-death process satisfies, for all $i\in\inter{0}{n}$, 
$$\pi(i) = \left(\binom{n-1}{i-1} \lambda^{-(i-1)} + \binom{n-1}{i} \lambda^{-i}\right)\pi(0).$$
so $\pi(0) \geq \frac 1 2 \left( 1+\frac 1
  \lambda \right)^{-n}$.

The bounding chain can be bounded by the following process: when in
state 0 only, vertex 0 can be removed with probability $1/(n+1)$ (all
events are active for removal, none for addition).  Then the expected
time $\tau_1$ for reaching a state $\langle B, D\rangle$ where
$0\notin D$ is (when $\lambda\geq n$)
$$\E{\tau_1} = \frac{(n+1)}{\pi(0)} \leq 2e(n+1).$$ 

For the second step, consider the birth-and-death process on
$\inter{0}{n}$ where state $i$ represents the 
$\langle B, D\rangle$ such that $|B|=n-i$ and $0\notin D$. For $i>
0$, $i$ vertices are active for addition and at least $n-i$ are active for
removal.
The transitions probabilities are exactly the probabilities
$p_{i,j}$ defined in Eq.~\eqref{eq:pi}.

Simple computations show that the hitting time $\tau_2$ of state $0$
from state $n$ satisfies $$\E{\tau_2} =
\left(\frac{1+\lambda}{\lambda}\right)^n + n \leq n +
e^{n/\lambda} = n+O(1).$$

Finally, note that in state $n$, vertex 0 is active for addition, and
in case this event is generated (which happens with probability $\frac{1}{n+1}$), we
have to take into account the return time from the first step ($0\in D$) to the second step ($0\notin D$). By the
Markov inequality, the probability that state $n$ is visited again
before state $0$ is at most $\frac{\pi(n)}{\pi(0)} = \lambda^{-n}$. 

As a consequence, the expected coupling time satisfies
$$\E{\tau^{\O}} \leq \E{\tau_1} + \E{\tau_2} + O(1)\leq (2e+1)n + O(1).$$
Notice that the coupling time does not depend on $\lambda$ and is linear
in $n$. It therefore does better then the other samplers presented above.

\subsection{Numerical experiments}
\label{sec:numeric}
We now do an experimental comparison of the three samplers described
in Section~\ref{ssec:sampleindep} for two models: the star graph, that
has been precisely analysed in Paragraph~\ref{ssec:star}, and the
Barab\'asi-Albert model \cite{AB02}.

\paragraph{Star graph}
We performed experiments for a star graph with 100 vertices and for
different values of $\lambda$. For each value of $\lambda$ and each
sampler, 1000 experiments have been performed, and the average number
of transitions computed is depicted in Figures~\ref{fig:etoile}.

 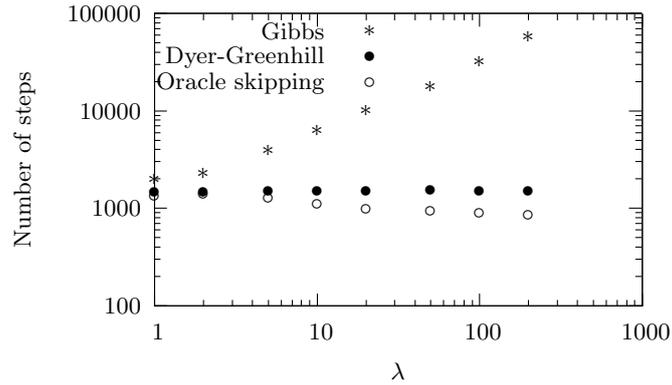
\begin{figure}[htbp]
  \centering
  \small
\setlength{\unitlength}{0.240900pt}
\ifx\plotpoint\undefined\newsavebox{\plotpoint}\fi
\sbox{\plotpoint}{\rule[-0.200pt]{0.400pt}{0.400pt}}%
\begin{picture}(1050,630)(0,0)
\sbox{\plotpoint}{\rule[-0.200pt]{0.400pt}{0.400pt}}%
\put(231.0,131.0){\rule[-0.200pt]{4.818pt}{0.400pt}}
\put(211,131){\makebox(0,0)[r]{ 100}}
\put(969.0,131.0){\rule[-0.200pt]{4.818pt}{0.400pt}}
\put(231.0,177.0){\rule[-0.200pt]{2.409pt}{0.400pt}}
\put(979.0,177.0){\rule[-0.200pt]{2.409pt}{0.400pt}}
\put(231.0,204.0){\rule[-0.200pt]{2.409pt}{0.400pt}}
\put(979.0,204.0){\rule[-0.200pt]{2.409pt}{0.400pt}}
\put(231.0,223.0){\rule[-0.200pt]{2.409pt}{0.400pt}}
\put(979.0,223.0){\rule[-0.200pt]{2.409pt}{0.400pt}}
\put(231.0,238.0){\rule[-0.200pt]{2.409pt}{0.400pt}}
\put(979.0,238.0){\rule[-0.200pt]{2.409pt}{0.400pt}}
\put(231.0,250.0){\rule[-0.200pt]{2.409pt}{0.400pt}}
\put(979.0,250.0){\rule[-0.200pt]{2.409pt}{0.400pt}}
\put(231.0,260.0){\rule[-0.200pt]{2.409pt}{0.400pt}}
\put(979.0,260.0){\rule[-0.200pt]{2.409pt}{0.400pt}}
\put(231.0,269.0){\rule[-0.200pt]{2.409pt}{0.400pt}}
\put(979.0,269.0){\rule[-0.200pt]{2.409pt}{0.400pt}}
\put(231.0,277.0){\rule[-0.200pt]{2.409pt}{0.400pt}}
\put(979.0,277.0){\rule[-0.200pt]{2.409pt}{0.400pt}}
\put(231.0,284.0){\rule[-0.200pt]{4.818pt}{0.400pt}}
\put(211,284){\makebox(0,0)[r]{ 1000}}
\put(969.0,284.0){\rule[-0.200pt]{4.818pt}{0.400pt}}
\put(231.0,330.0){\rule[-0.200pt]{2.409pt}{0.400pt}}
\put(979.0,330.0){\rule[-0.200pt]{2.409pt}{0.400pt}}
\put(231.0,357.0){\rule[-0.200pt]{2.409pt}{0.400pt}}
\put(979.0,357.0){\rule[-0.200pt]{2.409pt}{0.400pt}}
\put(231.0,376.0){\rule[-0.200pt]{2.409pt}{0.400pt}}
\put(979.0,376.0){\rule[-0.200pt]{2.409pt}{0.400pt}}
\put(231.0,390.0){\rule[-0.200pt]{2.409pt}{0.400pt}}
\put(979.0,390.0){\rule[-0.200pt]{2.409pt}{0.400pt}}
\put(231.0,402.0){\rule[-0.200pt]{2.409pt}{0.400pt}}
\put(979.0,402.0){\rule[-0.200pt]{2.409pt}{0.400pt}}
\put(231.0,413.0){\rule[-0.200pt]{2.409pt}{0.400pt}}
\put(979.0,413.0){\rule[-0.200pt]{2.409pt}{0.400pt}}
\put(231.0,422.0){\rule[-0.200pt]{2.409pt}{0.400pt}}
\put(979.0,422.0){\rule[-0.200pt]{2.409pt}{0.400pt}}
\put(231.0,429.0){\rule[-0.200pt]{2.409pt}{0.400pt}}
\put(979.0,429.0){\rule[-0.200pt]{2.409pt}{0.400pt}}
\put(231.0,436.0){\rule[-0.200pt]{4.818pt}{0.400pt}}
\put(211,436){\makebox(0,0)[r]{ 10000}}
\put(969.0,436.0){\rule[-0.200pt]{4.818pt}{0.400pt}}
\put(231.0,482.0){\rule[-0.200pt]{2.409pt}{0.400pt}}
\put(979.0,482.0){\rule[-0.200pt]{2.409pt}{0.400pt}}
\put(231.0,509.0){\rule[-0.200pt]{2.409pt}{0.400pt}}
\put(979.0,509.0){\rule[-0.200pt]{2.409pt}{0.400pt}}
\put(231.0,528.0){\rule[-0.200pt]{2.409pt}{0.400pt}}
\put(979.0,528.0){\rule[-0.200pt]{2.409pt}{0.400pt}}
\put(231.0,543.0){\rule[-0.200pt]{2.409pt}{0.400pt}}
\put(979.0,543.0){\rule[-0.200pt]{2.409pt}{0.400pt}}
\put(231.0,555.0){\rule[-0.200pt]{2.409pt}{0.400pt}}
\put(979.0,555.0){\rule[-0.200pt]{2.409pt}{0.400pt}}
\put(231.0,565.0){\rule[-0.200pt]{2.409pt}{0.400pt}}
\put(979.0,565.0){\rule[-0.200pt]{2.409pt}{0.400pt}}
\put(231.0,574.0){\rule[-0.200pt]{2.409pt}{0.400pt}}
\put(979.0,574.0){\rule[-0.200pt]{2.409pt}{0.400pt}}
\put(231.0,582.0){\rule[-0.200pt]{2.409pt}{0.400pt}}
\put(979.0,582.0){\rule[-0.200pt]{2.409pt}{0.400pt}}
\put(231.0,589.0){\rule[-0.200pt]{4.818pt}{0.400pt}}
\put(211,589){\makebox(0,0)[r]{ 100000}}
\put(969.0,589.0){\rule[-0.200pt]{4.818pt}{0.400pt}}
\put(231.0,131.0){\rule[-0.200pt]{0.400pt}{4.818pt}}
\put(231,90){\makebox(0,0){ 1}}
\put(231.0,569.0){\rule[-0.200pt]{0.400pt}{4.818pt}}
\put(307.0,131.0){\rule[-0.200pt]{0.400pt}{2.409pt}}
\put(307.0,579.0){\rule[-0.200pt]{0.400pt}{2.409pt}}
\put(352.0,131.0){\rule[-0.200pt]{0.400pt}{2.409pt}}
\put(352.0,579.0){\rule[-0.200pt]{0.400pt}{2.409pt}}
\put(383.0,131.0){\rule[-0.200pt]{0.400pt}{2.409pt}}
\put(383.0,579.0){\rule[-0.200pt]{0.400pt}{2.409pt}}
\put(408.0,131.0){\rule[-0.200pt]{0.400pt}{2.409pt}}
\put(408.0,579.0){\rule[-0.200pt]{0.400pt}{2.409pt}}
\put(428.0,131.0){\rule[-0.200pt]{0.400pt}{2.409pt}}
\put(428.0,579.0){\rule[-0.200pt]{0.400pt}{2.409pt}}
\put(445.0,131.0){\rule[-0.200pt]{0.400pt}{2.409pt}}
\put(445.0,579.0){\rule[-0.200pt]{0.400pt}{2.409pt}}
\put(459.0,131.0){\rule[-0.200pt]{0.400pt}{2.409pt}}
\put(459.0,579.0){\rule[-0.200pt]{0.400pt}{2.409pt}}
\put(472.0,131.0){\rule[-0.200pt]{0.400pt}{2.409pt}}
\put(472.0,579.0){\rule[-0.200pt]{0.400pt}{2.409pt}}
\put(484.0,131.0){\rule[-0.200pt]{0.400pt}{4.818pt}}
\put(484,90){\makebox(0,0){ 10}}
\put(484.0,569.0){\rule[-0.200pt]{0.400pt}{4.818pt}}
\put(560.0,131.0){\rule[-0.200pt]{0.400pt}{2.409pt}}
\put(560.0,579.0){\rule[-0.200pt]{0.400pt}{2.409pt}}
\put(604.0,131.0){\rule[-0.200pt]{0.400pt}{2.409pt}}
\put(604.0,579.0){\rule[-0.200pt]{0.400pt}{2.409pt}}
\put(636.0,131.0){\rule[-0.200pt]{0.400pt}{2.409pt}}
\put(636.0,579.0){\rule[-0.200pt]{0.400pt}{2.409pt}}
\put(660.0,131.0){\rule[-0.200pt]{0.400pt}{2.409pt}}
\put(660.0,579.0){\rule[-0.200pt]{0.400pt}{2.409pt}}
\put(680.0,131.0){\rule[-0.200pt]{0.400pt}{2.409pt}}
\put(680.0,579.0){\rule[-0.200pt]{0.400pt}{2.409pt}}
\put(697.0,131.0){\rule[-0.200pt]{0.400pt}{2.409pt}}
\put(697.0,579.0){\rule[-0.200pt]{0.400pt}{2.409pt}}
\put(712.0,131.0){\rule[-0.200pt]{0.400pt}{2.409pt}}
\put(712.0,579.0){\rule[-0.200pt]{0.400pt}{2.409pt}}
\put(725.0,131.0){\rule[-0.200pt]{0.400pt}{2.409pt}}
\put(725.0,579.0){\rule[-0.200pt]{0.400pt}{2.409pt}}
\put(736.0,131.0){\rule[-0.200pt]{0.400pt}{4.818pt}}
\put(736,90){\makebox(0,0){ 100}}
\put(736.0,569.0){\rule[-0.200pt]{0.400pt}{4.818pt}}
\put(812.0,131.0){\rule[-0.200pt]{0.400pt}{2.409pt}}
\put(812.0,579.0){\rule[-0.200pt]{0.400pt}{2.409pt}}
\put(857.0,131.0){\rule[-0.200pt]{0.400pt}{2.409pt}}
\put(857.0,579.0){\rule[-0.200pt]{0.400pt}{2.409pt}}
\put(888.0,131.0){\rule[-0.200pt]{0.400pt}{2.409pt}}
\put(888.0,579.0){\rule[-0.200pt]{0.400pt}{2.409pt}}
\put(913.0,131.0){\rule[-0.200pt]{0.400pt}{2.409pt}}
\put(913.0,579.0){\rule[-0.200pt]{0.400pt}{2.409pt}}
\put(933.0,131.0){\rule[-0.200pt]{0.400pt}{2.409pt}}
\put(933.0,579.0){\rule[-0.200pt]{0.400pt}{2.409pt}}
\put(950.0,131.0){\rule[-0.200pt]{0.400pt}{2.409pt}}
\put(950.0,579.0){\rule[-0.200pt]{0.400pt}{2.409pt}}
\put(965.0,131.0){\rule[-0.200pt]{0.400pt}{2.409pt}}
\put(965.0,579.0){\rule[-0.200pt]{0.400pt}{2.409pt}}
\put(977.0,131.0){\rule[-0.200pt]{0.400pt}{2.409pt}}
\put(977.0,579.0){\rule[-0.200pt]{0.400pt}{2.409pt}}
\put(989.0,131.0){\rule[-0.200pt]{0.400pt}{4.818pt}}
\put(989,90){\makebox(0,0){ 1000}}
\put(989.0,569.0){\rule[-0.200pt]{0.400pt}{4.818pt}}
\put(231.0,131.0){\rule[-0.200pt]{0.400pt}{110.332pt}}
\put(231.0,131.0){\rule[-0.200pt]{182.602pt}{0.400pt}}
\put(989.0,131.0){\rule[-0.200pt]{0.400pt}{110.332pt}}
\put(231.0,589.0){\rule[-0.200pt]{182.602pt}{0.400pt}}
\put(30,360){\makebox(0,0){\rotatebox{90}{Number of steps}}}
\put(610,29){\makebox(0,0){$\lambda$}}
\put(496,562){\makebox(0,0)[r]{Gibbs}}
\put(231,329){\makebox(0,0){$\ast$}}
\put(307,338){\makebox(0,0){$\ast$}}
\put(408,375){\makebox(0,0){$\ast$}}
\put(484,405){\makebox(0,0){$\ast$}}
\put(560,437){\makebox(0,0){$\ast$}}
\put(660,474){\makebox(0,0){$\ast$}}
\put(736,513){\makebox(0,0){$\ast$}}
\put(812,553){\makebox(0,0){$\ast$}}
\put(566,562){\makebox(0,0){$\ast$}}
\put(496,521){\makebox(0,0)[r]{Dyer-Greenhill}}
\put(231,309){\makebox(0,0){$\bullet$}}
\put(307,309){\makebox(0,0){$\bullet$}}
\put(408,310){\makebox(0,0){$\bullet$}}
\put(484,310){\makebox(0,0){$\bullet$}}
\put(560,310){\makebox(0,0){$\bullet$}}
\put(660,311){\makebox(0,0){$\bullet$}}
\put(736,310){\makebox(0,0){$\bullet$}}
\put(812,310){\makebox(0,0){$\bullet$}}
\put(566,521){\makebox(0,0){$\bullet$}}
\sbox{\plotpoint}{\rule[-0.400pt]{0.800pt}{0.800pt}}%
\sbox{\plotpoint}{\rule[-0.200pt]{0.400pt}{0.400pt}}%
\put(496,480){\makebox(0,0)[r]{Oracle skipping}}
\sbox{\plotpoint}{\rule[-0.400pt]{0.800pt}{0.800pt}}%
\put(231,302){\makebox(0,0){$\circ$}}
\put(307,305){\makebox(0,0){$\circ$}}
\put(408,299){\makebox(0,0){$\circ$}}
\put(484,290){\makebox(0,0){$\circ$}}
\put(560,282){\makebox(0,0){$\circ$}}
\put(660,279){\makebox(0,0){$\circ$}}
\put(736,275){\makebox(0,0){$\circ$}}
\put(812,273){\makebox(0,0){$\circ$}}
\put(566,480){\makebox(0,0){$\circ$}}
\sbox{\plotpoint}{\rule[-0.200pt]{0.400pt}{0.400pt}}%
\put(231.0,131.0){\rule[-0.200pt]{0.400pt}{110.332pt}}
\put(231.0,131.0){\rule[-0.200pt]{182.602pt}{0.400pt}}
\put(989.0,131.0){\rule[-0.200pt]{0.400pt}{110.332pt}}
\put(231.0,589.0){\rule[-0.200pt]{182.602pt}{0.400pt}}
\end{picture}
  \caption{Number of events generated by CFTP algortihms for the star graph with $100$ vertices for different values of $\lambda$.}
  \label{fig:etoile}
\end{figure}

The first remark is that both Dyer-Greenhill and oracle skipping
samplers outperform the Gibbs sampler. Second, the Dyer-Greenhill
sampler seams insensitive to the value of $\lambda$, which conforms to
the bound $n\ln n + O(1)$ given in Section~\ref{ssec:star}. Finally,
the oracle skipping scheme is always the most efficient algorithm. It
is noticeable that the number of event generated decreases with
$\lambda$. This can be explained the following way: large independent
sets are favored when $\lambda$ grows. Then, after reaching the
independent set $\inter{1}{n}$, whose probability grows with
$\lambda$, the probability that the next event is active is less than 
$1/(1+\lambda)$. As a consequence, many events are skipped. 

The difference in behavior between the Dyer-Greenhill and oracle skipping samplers is
more obvious with the star graph with 1000 vertices, as depicted in
Figure~\ref{fig:etoile1000} (100 experiments are run for each value of $\lambda$).

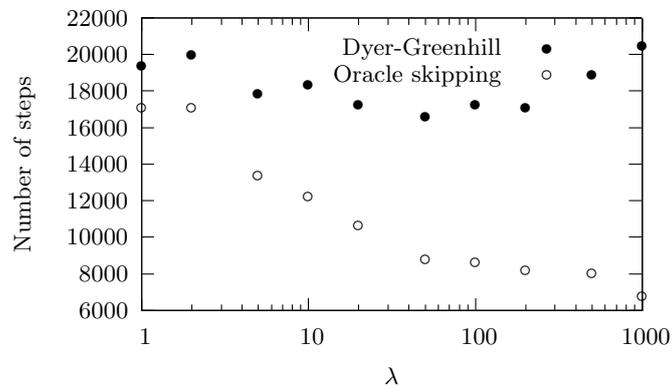
\begin{figure}[htbp]
  \centering
  \small
\setlength{\unitlength}{0.240900pt}
\ifx\plotpoint\undefined\newsavebox{\plotpoint}\fi
\sbox{\plotpoint}{\rule[-0.200pt]{0.400pt}{0.400pt}}%
\begin{picture}(1050,630)(0,0)
\sbox{\plotpoint}{\rule[-0.200pt]{0.400pt}{0.400pt}}%
\put(211.0,131.0){\rule[-0.200pt]{4.818pt}{0.400pt}}
\put(191,131){\makebox(0,0)[r]{ 6000}}
\put(969.0,131.0){\rule[-0.200pt]{4.818pt}{0.400pt}}
\put(211.0,188.0){\rule[-0.200pt]{4.818pt}{0.400pt}}
\put(191,188){\makebox(0,0)[r]{ 8000}}
\put(969.0,188.0){\rule[-0.200pt]{4.818pt}{0.400pt}}
\put(211.0,246.0){\rule[-0.200pt]{4.818pt}{0.400pt}}
\put(191,246){\makebox(0,0)[r]{ 10000}}
\put(969.0,246.0){\rule[-0.200pt]{4.818pt}{0.400pt}}
\put(211.0,303.0){\rule[-0.200pt]{4.818pt}{0.400pt}}
\put(191,303){\makebox(0,0)[r]{ 12000}}
\put(969.0,303.0){\rule[-0.200pt]{4.818pt}{0.400pt}}
\put(211.0,360.0){\rule[-0.200pt]{4.818pt}{0.400pt}}
\put(191,360){\makebox(0,0)[r]{ 14000}}
\put(969.0,360.0){\rule[-0.200pt]{4.818pt}{0.400pt}}
\put(211.0,417.0){\rule[-0.200pt]{4.818pt}{0.400pt}}
\put(191,417){\makebox(0,0)[r]{ 16000}}
\put(969.0,417.0){\rule[-0.200pt]{4.818pt}{0.400pt}}
\put(211.0,475.0){\rule[-0.200pt]{4.818pt}{0.400pt}}
\put(191,475){\makebox(0,0)[r]{ 18000}}
\put(969.0,475.0){\rule[-0.200pt]{4.818pt}{0.400pt}}
\put(211.0,532.0){\rule[-0.200pt]{4.818pt}{0.400pt}}
\put(191,532){\makebox(0,0)[r]{ 20000}}
\put(969.0,532.0){\rule[-0.200pt]{4.818pt}{0.400pt}}
\put(211.0,589.0){\rule[-0.200pt]{4.818pt}{0.400pt}}
\put(191,589){\makebox(0,0)[r]{ 22000}}
\put(969.0,589.0){\rule[-0.200pt]{4.818pt}{0.400pt}}
\put(211.0,131.0){\rule[-0.200pt]{0.400pt}{4.818pt}}
\put(211,90){\makebox(0,0){ 1}}
\put(211.0,569.0){\rule[-0.200pt]{0.400pt}{4.818pt}}
\put(289.0,131.0){\rule[-0.200pt]{0.400pt}{2.409pt}}
\put(289.0,579.0){\rule[-0.200pt]{0.400pt}{2.409pt}}
\put(335.0,131.0){\rule[-0.200pt]{0.400pt}{2.409pt}}
\put(335.0,579.0){\rule[-0.200pt]{0.400pt}{2.409pt}}
\put(367.0,131.0){\rule[-0.200pt]{0.400pt}{2.409pt}}
\put(367.0,579.0){\rule[-0.200pt]{0.400pt}{2.409pt}}
\put(392.0,131.0){\rule[-0.200pt]{0.400pt}{2.409pt}}
\put(392.0,579.0){\rule[-0.200pt]{0.400pt}{2.409pt}}
\put(413.0,131.0){\rule[-0.200pt]{0.400pt}{2.409pt}}
\put(413.0,579.0){\rule[-0.200pt]{0.400pt}{2.409pt}}
\put(430.0,131.0){\rule[-0.200pt]{0.400pt}{2.409pt}}
\put(430.0,579.0){\rule[-0.200pt]{0.400pt}{2.409pt}}
\put(445.0,131.0){\rule[-0.200pt]{0.400pt}{2.409pt}}
\put(445.0,579.0){\rule[-0.200pt]{0.400pt}{2.409pt}}
\put(458.0,131.0){\rule[-0.200pt]{0.400pt}{2.409pt}}
\put(458.0,579.0){\rule[-0.200pt]{0.400pt}{2.409pt}}
\put(470.0,131.0){\rule[-0.200pt]{0.400pt}{4.818pt}}
\put(470,90){\makebox(0,0){ 10}}
\put(470.0,569.0){\rule[-0.200pt]{0.400pt}{4.818pt}}
\put(548.0,131.0){\rule[-0.200pt]{0.400pt}{2.409pt}}
\put(548.0,579.0){\rule[-0.200pt]{0.400pt}{2.409pt}}
\put(594.0,131.0){\rule[-0.200pt]{0.400pt}{2.409pt}}
\put(594.0,579.0){\rule[-0.200pt]{0.400pt}{2.409pt}}
\put(626.0,131.0){\rule[-0.200pt]{0.400pt}{2.409pt}}
\put(626.0,579.0){\rule[-0.200pt]{0.400pt}{2.409pt}}
\put(652.0,131.0){\rule[-0.200pt]{0.400pt}{2.409pt}}
\put(652.0,579.0){\rule[-0.200pt]{0.400pt}{2.409pt}}
\put(672.0,131.0){\rule[-0.200pt]{0.400pt}{2.409pt}}
\put(672.0,579.0){\rule[-0.200pt]{0.400pt}{2.409pt}}
\put(689.0,131.0){\rule[-0.200pt]{0.400pt}{2.409pt}}
\put(689.0,579.0){\rule[-0.200pt]{0.400pt}{2.409pt}}
\put(705.0,131.0){\rule[-0.200pt]{0.400pt}{2.409pt}}
\put(705.0,579.0){\rule[-0.200pt]{0.400pt}{2.409pt}}
\put(718.0,131.0){\rule[-0.200pt]{0.400pt}{2.409pt}}
\put(718.0,579.0){\rule[-0.200pt]{0.400pt}{2.409pt}}
\put(730.0,131.0){\rule[-0.200pt]{0.400pt}{4.818pt}}
\put(730,90){\makebox(0,0){ 100}}
\put(730.0,569.0){\rule[-0.200pt]{0.400pt}{4.818pt}}
\put(808.0,131.0){\rule[-0.200pt]{0.400pt}{2.409pt}}
\put(808.0,579.0){\rule[-0.200pt]{0.400pt}{2.409pt}}
\put(853.0,131.0){\rule[-0.200pt]{0.400pt}{2.409pt}}
\put(853.0,579.0){\rule[-0.200pt]{0.400pt}{2.409pt}}
\put(886.0,131.0){\rule[-0.200pt]{0.400pt}{2.409pt}}
\put(886.0,579.0){\rule[-0.200pt]{0.400pt}{2.409pt}}
\put(911.0,131.0){\rule[-0.200pt]{0.400pt}{2.409pt}}
\put(911.0,579.0){\rule[-0.200pt]{0.400pt}{2.409pt}}
\put(931.0,131.0){\rule[-0.200pt]{0.400pt}{2.409pt}}
\put(931.0,579.0){\rule[-0.200pt]{0.400pt}{2.409pt}}
\put(949.0,131.0){\rule[-0.200pt]{0.400pt}{2.409pt}}
\put(949.0,579.0){\rule[-0.200pt]{0.400pt}{2.409pt}}
\put(964.0,131.0){\rule[-0.200pt]{0.400pt}{2.409pt}}
\put(964.0,579.0){\rule[-0.200pt]{0.400pt}{2.409pt}}
\put(977.0,131.0){\rule[-0.200pt]{0.400pt}{2.409pt}}
\put(977.0,579.0){\rule[-0.200pt]{0.400pt}{2.409pt}}
\put(989.0,131.0){\rule[-0.200pt]{0.400pt}{4.818pt}}
\put(989,90){\makebox(0,0){ 1000}}
\put(989.0,569.0){\rule[-0.200pt]{0.400pt}{4.818pt}}
\put(211.0,131.0){\rule[-0.200pt]{0.400pt}{110.332pt}}
\put(211.0,131.0){\rule[-0.200pt]{187.420pt}{0.400pt}}
\put(989.0,131.0){\rule[-0.200pt]{0.400pt}{110.332pt}}
\put(211.0,589.0){\rule[-0.200pt]{187.420pt}{0.400pt}}
\put(30,360){\makebox(0,0){\rotatebox{90}{Number of steps}}}
\put(600,29){\makebox(0,0){$\lambda$}}
\put(771,540){\makebox(0,0)[r]{Dyer-Greenhill}}
\put(211,513){\makebox(0,0){$\bullet$}}
\put(289,531){\makebox(0,0){$\bullet$}}
\put(392,470){\makebox(0,0){$\bullet$}}
\put(470,484){\makebox(0,0){$\bullet$}}
\put(548,452){\makebox(0,0){$\bullet$}}
\put(652,433){\makebox(0,0){$\bullet$}}
\put(730,452){\makebox(0,0){$\bullet$}}
\put(808,447){\makebox(0,0){$\bullet$}}
\put(911,499){\makebox(0,0){$\bullet$}}
\put(989,545){\makebox(0,0){$\bullet$}}
\put(841,540){\makebox(0,0){$\bullet$}}
\put(771,499){\makebox(0,0)[r]{Oracle skipping}}
\put(211,448){\makebox(0,0){$\circ$}}
\put(289,447){\makebox(0,0){$\circ$}}
\put(392,342){\makebox(0,0){$\circ$}}
\put(470,308){\makebox(0,0){$\circ$}}
\put(548,263){\makebox(0,0){$\circ$}}
\put(652,210){\makebox(0,0){$\circ$}}
\put(730,206){\makebox(0,0){$\circ$}}
\put(808,193){\makebox(0,0){$\circ$}}
\put(911,188){\makebox(0,0){$\circ$}}
\put(989,153){\makebox(0,0){$\circ$}}
\put(841,499){\makebox(0,0){$\circ$}}
\put(211.0,131.0){\rule[-0.200pt]{0.400pt}{110.332pt}}
\put(211.0,131.0){\rule[-0.200pt]{187.420pt}{0.400pt}}
\put(989.0,131.0){\rule[-0.200pt]{0.400pt}{110.332pt}}
\put(211.0,589.0){\rule[-0.200pt]{187.420pt}{0.400pt}}
\end{picture}
  \caption{Number of events generated by CFTP algortihms for the star graph with $1000$ vertices for different values of $\lambda$.}
  \label{fig:etoile1000}
\end{figure}

\paragraph{Barab\'asi-Albert model}
We now generate a random graph with preferential attachment.  Start
from a clique with 5 vertices and at each step add one new vertex $v$
and two edges $(v,w_1)$ and $(v,w_2)$, where $w_1$ and $w_2$ are
chosen at random with probability proportional to their
degree. Figure~\ref{fig:barabasi} compares the average number of
events generated for 100 experiments with the three samplers, for
graphs with 100 vertices.

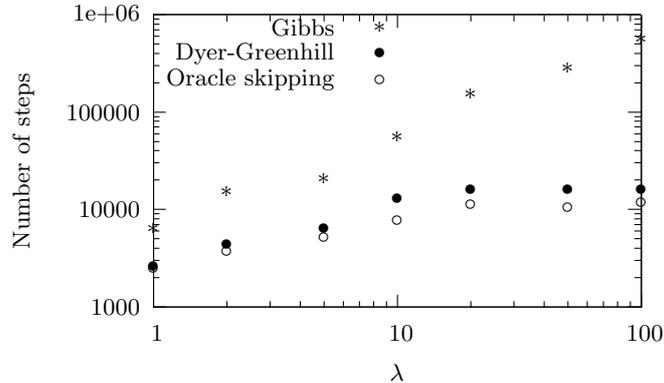
\begin{figure}[htbp]
  \centering
  \small
\setlength{\unitlength}{0.240900pt}
\ifx\plotpoint\undefined\newsavebox{\plotpoint}\fi
\sbox{\plotpoint}{\rule[-0.200pt]{0.400pt}{0.400pt}}%
\begin{picture}(1050,630)(0,0)
\sbox{\plotpoint}{\rule[-0.200pt]{0.400pt}{0.400pt}}%
\put(231.0,131.0){\rule[-0.200pt]{4.818pt}{0.400pt}}
\put(211,131){\makebox(0,0)[r]{ 1000}}
\put(969.0,131.0){\rule[-0.200pt]{4.818pt}{0.400pt}}
\put(231.0,177.0){\rule[-0.200pt]{2.409pt}{0.400pt}}
\put(979.0,177.0){\rule[-0.200pt]{2.409pt}{0.400pt}}
\put(231.0,204.0){\rule[-0.200pt]{2.409pt}{0.400pt}}
\put(979.0,204.0){\rule[-0.200pt]{2.409pt}{0.400pt}}
\put(231.0,223.0){\rule[-0.200pt]{2.409pt}{0.400pt}}
\put(979.0,223.0){\rule[-0.200pt]{2.409pt}{0.400pt}}
\put(231.0,238.0){\rule[-0.200pt]{2.409pt}{0.400pt}}
\put(979.0,238.0){\rule[-0.200pt]{2.409pt}{0.400pt}}
\put(231.0,250.0){\rule[-0.200pt]{2.409pt}{0.400pt}}
\put(979.0,250.0){\rule[-0.200pt]{2.409pt}{0.400pt}}
\put(231.0,260.0){\rule[-0.200pt]{2.409pt}{0.400pt}}
\put(979.0,260.0){\rule[-0.200pt]{2.409pt}{0.400pt}}
\put(231.0,269.0){\rule[-0.200pt]{2.409pt}{0.400pt}}
\put(979.0,269.0){\rule[-0.200pt]{2.409pt}{0.400pt}}
\put(231.0,277.0){\rule[-0.200pt]{2.409pt}{0.400pt}}
\put(979.0,277.0){\rule[-0.200pt]{2.409pt}{0.400pt}}
\put(231.0,284.0){\rule[-0.200pt]{4.818pt}{0.400pt}}
\put(211,284){\makebox(0,0)[r]{ 10000}}
\put(969.0,284.0){\rule[-0.200pt]{4.818pt}{0.400pt}}
\put(231.0,330.0){\rule[-0.200pt]{2.409pt}{0.400pt}}
\put(979.0,330.0){\rule[-0.200pt]{2.409pt}{0.400pt}}
\put(231.0,357.0){\rule[-0.200pt]{2.409pt}{0.400pt}}
\put(979.0,357.0){\rule[-0.200pt]{2.409pt}{0.400pt}}
\put(231.0,376.0){\rule[-0.200pt]{2.409pt}{0.400pt}}
\put(979.0,376.0){\rule[-0.200pt]{2.409pt}{0.400pt}}
\put(231.0,390.0){\rule[-0.200pt]{2.409pt}{0.400pt}}
\put(979.0,390.0){\rule[-0.200pt]{2.409pt}{0.400pt}}
\put(231.0,402.0){\rule[-0.200pt]{2.409pt}{0.400pt}}
\put(979.0,402.0){\rule[-0.200pt]{2.409pt}{0.400pt}}
\put(231.0,413.0){\rule[-0.200pt]{2.409pt}{0.400pt}}
\put(979.0,413.0){\rule[-0.200pt]{2.409pt}{0.400pt}}
\put(231.0,422.0){\rule[-0.200pt]{2.409pt}{0.400pt}}
\put(979.0,422.0){\rule[-0.200pt]{2.409pt}{0.400pt}}
\put(231.0,429.0){\rule[-0.200pt]{2.409pt}{0.400pt}}
\put(979.0,429.0){\rule[-0.200pt]{2.409pt}{0.400pt}}
\put(231.0,436.0){\rule[-0.200pt]{4.818pt}{0.400pt}}
\put(211,436){\makebox(0,0)[r]{ 100000}}
\put(969.0,436.0){\rule[-0.200pt]{4.818pt}{0.400pt}}
\put(231.0,482.0){\rule[-0.200pt]{2.409pt}{0.400pt}}
\put(979.0,482.0){\rule[-0.200pt]{2.409pt}{0.400pt}}
\put(231.0,509.0){\rule[-0.200pt]{2.409pt}{0.400pt}}
\put(979.0,509.0){\rule[-0.200pt]{2.409pt}{0.400pt}}
\put(231.0,528.0){\rule[-0.200pt]{2.409pt}{0.400pt}}
\put(979.0,528.0){\rule[-0.200pt]{2.409pt}{0.400pt}}
\put(231.0,543.0){\rule[-0.200pt]{2.409pt}{0.400pt}}
\put(979.0,543.0){\rule[-0.200pt]{2.409pt}{0.400pt}}
\put(231.0,555.0){\rule[-0.200pt]{2.409pt}{0.400pt}}
\put(979.0,555.0){\rule[-0.200pt]{2.409pt}{0.400pt}}
\put(231.0,565.0){\rule[-0.200pt]{2.409pt}{0.400pt}}
\put(979.0,565.0){\rule[-0.200pt]{2.409pt}{0.400pt}}
\put(231.0,574.0){\rule[-0.200pt]{2.409pt}{0.400pt}}
\put(979.0,574.0){\rule[-0.200pt]{2.409pt}{0.400pt}}
\put(231.0,582.0){\rule[-0.200pt]{2.409pt}{0.400pt}}
\put(979.0,582.0){\rule[-0.200pt]{2.409pt}{0.400pt}}
\put(231.0,589.0){\rule[-0.200pt]{4.818pt}{0.400pt}}
\put(211,589){\makebox(0,0)[r]{ 1e+06}}
\put(969.0,589.0){\rule[-0.200pt]{4.818pt}{0.400pt}}
\put(231.0,131.0){\rule[-0.200pt]{0.400pt}{4.818pt}}
\put(231,90){\makebox(0,0){ 1}}
\put(231.0,569.0){\rule[-0.200pt]{0.400pt}{4.818pt}}
\put(345.0,131.0){\rule[-0.200pt]{0.400pt}{2.409pt}}
\put(345.0,579.0){\rule[-0.200pt]{0.400pt}{2.409pt}}
\put(412.0,131.0){\rule[-0.200pt]{0.400pt}{2.409pt}}
\put(412.0,579.0){\rule[-0.200pt]{0.400pt}{2.409pt}}
\put(459.0,131.0){\rule[-0.200pt]{0.400pt}{2.409pt}}
\put(459.0,579.0){\rule[-0.200pt]{0.400pt}{2.409pt}}
\put(496.0,131.0){\rule[-0.200pt]{0.400pt}{2.409pt}}
\put(496.0,579.0){\rule[-0.200pt]{0.400pt}{2.409pt}}
\put(526.0,131.0){\rule[-0.200pt]{0.400pt}{2.409pt}}
\put(526.0,579.0){\rule[-0.200pt]{0.400pt}{2.409pt}}
\put(551.0,131.0){\rule[-0.200pt]{0.400pt}{2.409pt}}
\put(551.0,579.0){\rule[-0.200pt]{0.400pt}{2.409pt}}
\put(573.0,131.0){\rule[-0.200pt]{0.400pt}{2.409pt}}
\put(573.0,579.0){\rule[-0.200pt]{0.400pt}{2.409pt}}
\put(593.0,131.0){\rule[-0.200pt]{0.400pt}{2.409pt}}
\put(593.0,579.0){\rule[-0.200pt]{0.400pt}{2.409pt}}
\put(610.0,131.0){\rule[-0.200pt]{0.400pt}{4.818pt}}
\put(610,90){\makebox(0,0){ 10}}
\put(610.0,569.0){\rule[-0.200pt]{0.400pt}{4.818pt}}
\put(724.0,131.0){\rule[-0.200pt]{0.400pt}{2.409pt}}
\put(724.0,579.0){\rule[-0.200pt]{0.400pt}{2.409pt}}
\put(791.0,131.0){\rule[-0.200pt]{0.400pt}{2.409pt}}
\put(791.0,579.0){\rule[-0.200pt]{0.400pt}{2.409pt}}
\put(838.0,131.0){\rule[-0.200pt]{0.400pt}{2.409pt}}
\put(838.0,579.0){\rule[-0.200pt]{0.400pt}{2.409pt}}
\put(875.0,131.0){\rule[-0.200pt]{0.400pt}{2.409pt}}
\put(875.0,579.0){\rule[-0.200pt]{0.400pt}{2.409pt}}
\put(905.0,131.0){\rule[-0.200pt]{0.400pt}{2.409pt}}
\put(905.0,579.0){\rule[-0.200pt]{0.400pt}{2.409pt}}
\put(930.0,131.0){\rule[-0.200pt]{0.400pt}{2.409pt}}
\put(930.0,579.0){\rule[-0.200pt]{0.400pt}{2.409pt}}
\put(952.0,131.0){\rule[-0.200pt]{0.400pt}{2.409pt}}
\put(952.0,579.0){\rule[-0.200pt]{0.400pt}{2.409pt}}
\put(972.0,131.0){\rule[-0.200pt]{0.400pt}{2.409pt}}
\put(972.0,579.0){\rule[-0.200pt]{0.400pt}{2.409pt}}
\put(989.0,131.0){\rule[-0.200pt]{0.400pt}{4.818pt}}
\put(989,90){\makebox(0,0){ 100}}
\put(989.0,569.0){\rule[-0.200pt]{0.400pt}{4.818pt}}
\put(231.0,131.0){\rule[-0.200pt]{0.400pt}{110.332pt}}
\put(231.0,131.0){\rule[-0.200pt]{182.602pt}{0.400pt}}
\put(989.0,131.0){\rule[-0.200pt]{0.400pt}{110.332pt}}
\put(231.0,589.0){\rule[-0.200pt]{182.602pt}{0.400pt}}
\put(30,360){\makebox(0,0){\rotatebox{90}{Number of steps}}}
\put(610,29){\makebox(0,0){$\lambda$}}
\put(513,569){\makebox(0,0)[r]{Gibbs}}
\put(231,254){\makebox(0,0){$\ast$}}
\put(345,312){\makebox(0,0){$\ast$}}
\put(496,332){\makebox(0,0){$\ast$}}
\put(610,397){\makebox(0,0){$\ast$}}
\put(724,464){\makebox(0,0){$\ast$}}
\put(875,505){\makebox(0,0){$\ast$}}
\put(989,550){\makebox(0,0){$\ast$}}
\put(583,569){\makebox(0,0){$\ast$}}
\put(513,528){\makebox(0,0)[r]{Dyer-Greenhill}}
\put(231,195){\makebox(0,0){$\bullet$}}
\put(345,229){\makebox(0,0){$\bullet$}}
\put(496,254){\makebox(0,0){$\bullet$}}
\put(610,301){\makebox(0,0){$\bullet$}}
\put(724,314){\makebox(0,0){$\bullet$}}
\put(875,314){\makebox(0,0){$\bullet$}}
\put(989,314){\makebox(0,0){$\bullet$}}
\put(583,528){\makebox(0,0){$\bullet$}}
\sbox{\plotpoint}{\rule[-0.400pt]{0.800pt}{0.800pt}}%
\sbox{\plotpoint}{\rule[-0.200pt]{0.400pt}{0.400pt}}%
\put(513,487){\makebox(0,0)[r]{Oracle skipping}}
\sbox{\plotpoint}{\rule[-0.400pt]{0.800pt}{0.800pt}}%
\put(231,191){\makebox(0,0){$\circ$}}
\put(345,218){\makebox(0,0){$\circ$}}
\put(496,239){\makebox(0,0){$\circ$}}
\put(610,266){\makebox(0,0){$\circ$}}
\put(724,291){\makebox(0,0){$\circ$}}
\put(875,286){\makebox(0,0){$\circ$}}
\put(989,294){\makebox(0,0){$\circ$}}
\put(583,487){\makebox(0,0){$\circ$}}
\sbox{\plotpoint}{\rule[-0.200pt]{0.400pt}{0.400pt}}%
\put(231.0,131.0){\rule[-0.200pt]{0.400pt}{110.332pt}}
\put(231.0,131.0){\rule[-0.200pt]{182.602pt}{0.400pt}}
\put(989.0,131.0){\rule[-0.200pt]{0.400pt}{110.332pt}}
\put(231.0,589.0){\rule[-0.200pt]{182.602pt}{0.400pt}}
\end{picture}
  \caption{Number of events generated by CFTP algortihms for the Barabasi-Albert model with $100$ vertices for different values of $\lambda$.}
  \label{fig:barabasi}
\end{figure}

Similarly to the star graph, Dyer-Greenhill and oracle skipping
samplers outperform the Gibbs sampler, and the oracle skipping sample
is sensitively better than the Dyer-Greenhill one. For large values,
those two samplers are not sensitive to $\lambda$ (or slightly improve
when $\lambda$ grows).

\bigskip

\section{Conclusions}
The main contribution of the paper is Algorithm 3, that speeds up the Markovian dynamics in the CFTP scheme for exact sampling from the stationary distribution of a Markov chain. 
We illustrated it here on the problem of randomly generating independent sets, but its applicability is much broader, within the context of the random generation of combinatorial objects using Glauber dynamics, or elsewhere. The application to the simulation of queueing networks is an ongoing research.


\bibliographystyle{abbrv}
\bibliography{biblio}

\end{document}